\documentclass[lettersize,journal]{IEEEtran}
\usepackage{amsmath,amsfonts}
\usepackage{algorithmic}
\usepackage{algorithm}
\usepackage{array}
\usepackage[caption=false,font=normalsize,labelfont=sf,textfont=sf]{subfig}
\usepackage{textcomp}
\usepackage{stfloats}
\usepackage{url}
\usepackage{verbatim}
\usepackage{graphicx}
\usepackage{cite}
\usepackage{amssymb}
\newtheorem{lemma}{Lemma}
\newcommand{\diff}{\mathop{}\!\mathrm{d}} 
\newtheorem{theorem}{Theorem}
\begin{document}

\title{Virtual Polarization Modulation: Enabling CSI-Free DCO-OFDM over Dynamic OWC Channels}

\author{Tian~Cao,~\IEEEmembership{Member,~IEEE},
	Ping~Wang,
	Tianfeng~Wu,~\IEEEmembership{Graduate~Student~Member,~IEEE},
	Kaile~Wang,
	and~Jian~Song,~\IEEEmembership{Fellow,~IEEE}
	\thanks{This work was supported by the National Natural Science Foundation of China under Grant 62401433 and 62071365, ``XiaoZhaoGongYong'' High-level Innovation and Entrepreneurship Talent Program of Shaanxi Province under Grant H024010008 and H025010001, Natural Science Basic Research Program of Shaanxi under Grant 2025JC-QYCX-054. (\it Corresponding author: Ping Wang, Jian Song.)}
	\thanks{Tian Cao, Ping Wang, and Kaile Wang are with School of Telecommunications Engineering, Xidian University, Xi’an 710071, China (E-mail: caotian@xidian.edu.cn, pingwang@xidian.edu.cn, wangkaile@xidian.edu.cn).}
	\thanks{Tianfeng Wu is with Beijing National Research Center for Information Science and Technology (BNRist), Department of Electronic Engineering, Tsinghua University, Beijing 100084, China (E-mail: wtf22@mails.tsinghua.edu.cn).}
	\thanks{Jian Song is with the Shenzhen International Graduate School, Tsinghua University, Shenzhen 518055, China (e-mail: jsong@tsinghua.edu.cn).}
	\thanks{This work has been submitted to the IEEE for possible publication. Copyright may be transferred without notice, after which this version may no longer be accessible.}
}



\maketitle

\begin{abstract}
In dynamically varying optical wireless communication (OWC) links, conventional quadrature amplitude modulation (QAM) in optical orthogonal frequency-division multiplexing (OFDM) requires frequent channel estimation and equalization, incurring pilot overhead and processing latency. This paper proposes a virtual polarization modulation (VPM)-based direct-current-biased optical OFDM (DCO-OFDM) scheme that maps each data symbol onto the three-dimensional Stokes space and places its corresponding Jones vector across two adjacent OFDM subcarriers. Using a rotation-based analytical framework, closed-form symbol error rate (SER) expressions are derived for arbitrary spherical constellations, along with upper and lower bounds and high signal-to-noise ratio (SNR) approximations. The framework is further extended to practical OWC scenarios with frequency-selective channels and atmospheric turbulence. Monte Carlo (MC) simulations validate the theoretical results. The results show that under practical OWC impairments, VPM outperforms QAM with least-squares (LS) channel estimation and minimum mean square error (MMSE) equalization. At a target SER of $10^{-5}$, 16-VPM achieves SNR gains of approximately $7.5$ dB and $4$ dB over equalized 16-QAM and 8-QAM, respectively, in frequency-selective channels, and a $6$ dB advantage over equalized 16-QAM under atmospheric turbulence. By eliminating the need for channel state information, the proposed VPM-based DCO-OFDM provides a robust and low-latency solution for dynamic OWC links.  
\end{abstract}


\section{Introduction}
The forthcoming sixth-generation (6G) wireless networks, guided by the IMT-2030 framework, is built upon four core design principles: \emph{connecting the unconnected, ubiquitous intelligence, security and resilience, and sustainability}\cite{Hossain_TTS_2025}. Under this vision, the traditional three usage cases defined in IMT-2020 are further evolved into six representative network scenarios, with particular emphasis on hyper-reliable and low-latency communication (HRLLC) to support emerging delay-sensitive applications\cite{Hong_JSAC_2026}. However, the severe congestion of the traditional sub-6 GHz radio-frequency (RF) spectrum has become a fundamental bottleneck in realizing these ambitious visions. Consequently, exploring higher frequency bands has become imperative. Among the potential solutions, optical wireless communication (OWC) has emerged as a compelling candidate \cite{Chu_JSAC_2026}. Operating across the infrared, visible, and ultraviolet spectra, OWC provides ultra-wide and license-free bandwidth for a wide range of communication scenarios, including indoor and outdoor visible light communication\cite{Kafizov_TWC_2024,Sharda_2022_TCom}, terrestrial free-space optical links\cite{Zedini_TWC_FSO2026}, satellite laser communications\cite{He_OE_2026}, underwater optical networks\cite{Li_2024_TVT}, and ultraviolet scattering communications\cite{Wu_JSAC_2025,Wang_WC_2026}.

In practical OWC systems, intensity modulation and direct detection (IM/DD) architectures are widely adopted due to their relatively low hardware complexity and cost compared with coherent optical transceivers\cite{Chaaban_CST_2021}. Under the inherent real-valued and non-negative signal constraints imposed by IM/DD, direct-current-biased optical orthogonal frequency division multiplexing (DCO-OFDM) has been widely used  because it effectively mitigates severe inter-symbol interference  while maintaining high spectral efficiency \cite{Ghassemlooy_book}. Most existing implementations employ baseband  in-phase/quadrature (IQ) modulation formats, such as $M$-ary quadrature amplitude modulation ($M$-QAM) and $M$-ary  phase-shift keying ($M$-PSK), which require accurate channel state information (CSI) at the receiver for frequency-domain equalization. However, practical OWC links often suffer from several challenging physical-layer impairments. Indoor optical channels typically exhibit multipath propagation and pronounced low-pass attenuation, primarily caused by optical reflections and the limited modulation bandwidth of commercially available light-emitting diodes (LEDs) \cite{Uysal_CM_2017}. In contrast, outdoor optical links, especially for long-range transmission, are highly sensitive to atmospheric turbulence, which induces serious irradiance scintillation and occasional deep fading events \cite{andrews2005laser}. These impairments degrade the reliability of conventional DCO-OFDM transmission.

Conventionally, CSI is acquired via pilot-assisted channel estimation (CE) methods \cite{cho2010mimo}. However, in scenarios with rapidly varying channels, reliable CE requires frequent pilot insertion. This incurs additional spectral and energy overhead and increases processing latency, posing a challenge for applications like vehicular OWC networks and potentially limiting the ability to meet HRLLC requirements envisioned for next-generation networks. Blind and learning-based CE methods have been proposed to reduce pilot overhead \cite{Chen_WCL_2021,openJ_2025}, but they generally incur high computational complexity at the receiver. This limits their practicality for resource-constrained optical edge devices, including lightweight internet-of-things sensors, unmanned aerial vehicles, and augmented reality headsets. Moreover, when the OWC channel experiences severe low-pass distortion or deep fading, the accuracy of CE degrades. The resulting CSI errors induce phase ambiguity and impair symbol detection, ultimately limiting system reliability. These challenges highlight the need for robust, low-complexity, and CSI-independent waveform designs capable of supporting reliable OWC transmission in dynamic environments.

To break the  reliance on instantaneous CSI, we draw an inspiration from the physical polarization state of electromagnetic waves. Polarization modulation encodes information into the relative amplitude and phase differences, which is mathematically represented by the Stokes parameters, between two orthogonal polarization states \cite{Guo_2017_CST, Henarejos_2018_Tcom}. If these two orthogonal physical channels experience similar influence, their common channel impairments could be mitigated in the Stokes space. Using this interesting property, some works  exploited polarization modulation to  mitigate phase noise in OFDM \cite{RN3364} and combat power amplifier  nonlinearities \cite{Wei_2013_CL}. 

Despite its theoretical advantages, practical polarization modulation in IM/DD-based OWC systems faces several limitations. From a hardware perspective, polarization modulation requires components such as optical polarizers, which increase system complexity and size. Moreover, light-emitting diodes (LEDs), widely used in practical OWC systems, emit unpolarized light. Generating polarized signals therefore requires external polarizers, which discard more than half of the optical power and reduce energy efficiency. From a channel perspective, maintaining orthogonal polarization states typically relies on strict line-of-sight alignment. In practical OWC environments, however, transceiver motion and multipath transmission often cause depolarization\cite{Guo_2017_CST}, which limits the effectiveness of polarization modulation in dynamic scenarios.

Therefore, we extend the polarization concept to the frequency domain. In IM/DD OWC systems, transmitted signals are real-valued and non-negative and do not rely on local oscillators, making them largely insensitive to frequency offset and Doppler effects \cite{Kinani_CST_2018}. Meanwhile, the low-pass characteristics of optoelectronic devices and OWC channel lead to strong correlation between adjacent subcarriers in DCO-OFDM\cite{Ghassemlooy_book}. By pairing neighboring subcarriers to form a ``virtual'' polarization state, channel-induced amplitude and phase distortions can be largely canceled in the Stokes space without requiring dual-polarized hardware.

Very recently, a related concept termed block modulation was introduced in RF systems to mitigate power amplifier nonlinearities and suppress out-of-band emissions by exploiting power-sum constraints among grouped subcarriers\cite{Fan_2026_TCom}. However, this design primarily targets RF hardware impairments, and a rigorous theoretical characterization of the symbol error rate (SER) is still lacking.

Motivated by these observations, this paper develops a CSI-free virtual polarization modulation (VPM) DCO-OFDM framework for dynamically varying OWC systems. By exploiting the  frequency-domain correlation between adjacent subcarriers, the proposed scheme removes the need for CSI. The main contributions are summarized as follows:

\begin{itemize}
	
	\item 
	We propose a VPM-based DCO-OFDM architecture for dynamic OWC channels. By mapping symbols into a three-dimensional Stokes space across two adjacent subcarriers, channel-induced amplitude and phase distortions are mitigated, eliminating the need for CSI and pilot symbols while reducing receiver processing latency.
	
	\item 
	Using a rotation-based analytical approach, we derive exact closed-form SER expressions for arbitrary spherical constellations over the additive white Gaussian noise (AWGN) channel. Tight upper and lower bounds, as well as high-SNR asymptotic approximations, are also obtained.

	\item 
	In practical OWC scenarios, the SER expressions are further derived under frequency-selective channels and atmospheric turbulence modeled by log-normal (LN) and Gamma-Gamma (GG) fading. In addition, the angular drift caused by channel distortion and noise is analyzed to provide further insight into system impairments.

\end{itemize}

The remainder of this paper is organized as follows. Section~II introduces the system model of the proposed VPM-based DCO-OFDM scheme. Section~III presents the analytical framework and derives the closed-form SER expressions, bounds, and asymptotic approximations. Section~IV extends the analysis to frequency-selective channels and atmospheric turbulence. Section~V provides results and discussions. Finally, Section~VI concludes the paper.

\textit{Notations:} Throughout this paper, boldface  letters denote  vectors or matrices. $\mathbf{I}_n$ denotes the $n \times n$ identity matrix. $(\cdot)^T$, $(\cdot)^H$, and $(\cdot)^*$ represent the transpose, Hermitian transpose, and complex conjugate operations, respectively. $\mathbb{E}[\cdot]$ denotes the expectation operator. $|\cdot|$ represents the absolute value of a scalar, and $\|\cdot\|$ denotes the Euclidean norm of a vector. $\Re\{\cdot\}$ and $\Im\{\cdot\}$ denote the real and imaginary parts of a complex number, respectively.

\section{System Model and Principle of VPM}
In this section, we present the overall architecture of the proposed VPM-based DCO-OFDM system. The VPM scheme acts as a CSI-free block modulation technique by exploiting the relationship between the Stokes and Jones vectors. The implementation details of the VPM scheme are provided below.

\subsection{VPM-Based Optical OFDM System}
We consider an IM/DD OWC system based on the DCO-OFDM\cite{Dissanayake_2013_JLT}. The overall block diagram of the VPM-based DCO-OFDM transceiver is illustrated in Fig.~\ref{fig:system_model}.

 \begin{figure}[htbp]
	 \centering
	 \includegraphics[width=0.98\columnwidth]{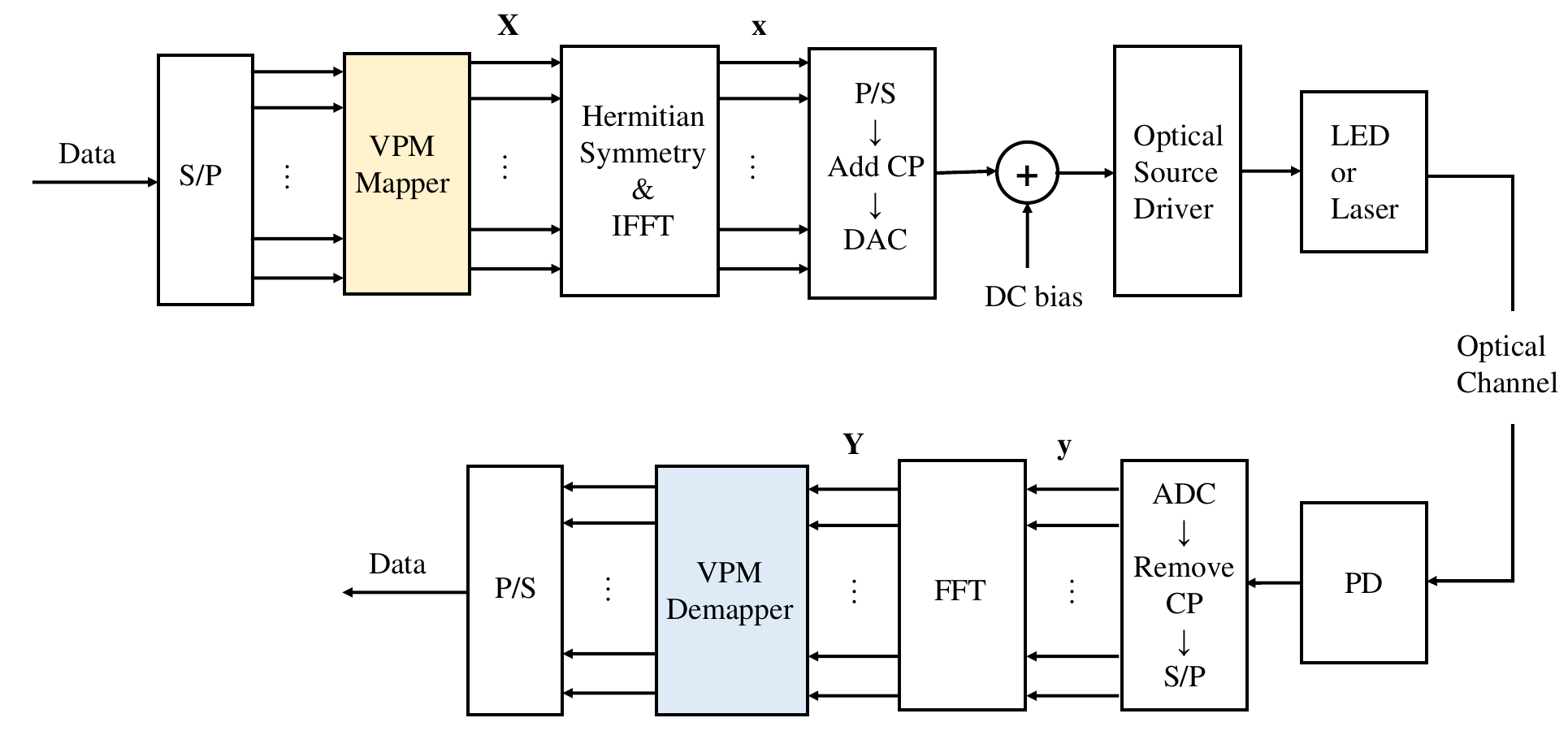}
	 \caption{Block diagram of the VPM-based DCO-OFDM system.}
	 \label{fig:system_model}
\end{figure}

At the transmitter, the incoming data are first converted into parallel symbols. Unlike conventional modulation schemes where each symbol is mapped to a single subcarrier, the proposed VPM mapper maps each symbol to a pair of complex-valued symbols forming a VPM block. For the $m$-th symbol, the mapper generates a Jones vector $\mathbf{E}_m = [E_{x,m}, E_{y,m}]^T$, which represents a virtual polarization state. The two components are then assigned to two adjacent  subcarriers of the OFDM symbol.

Let $\mathbf{X} = [X[0], X[1], \dots, X[N-1]]^T$ denote the frequency-domain OFDM symbol with $N$ subcarriers. To ensure a real-valued time-domain signal for IM/DD transmission, Hermitian symmetry is imposed on $\mathbf{X}$. The DC subcarrier ($k=0$) and Nyquist subcarrier ($k=N/2$) are set to zero, and the data subcarriers occupy the first half of the spectrum ($1 \le k \le N/2-1$).

Let the number of valid VPM blocks per OFDM symbol be $N_v = \lfloor (N/2 - 1) / 2 \rfloor$. The $m$-th block $\mathbf{E}_m$ is mapped to the data subcarriers as
\begin{equation}
	\begin{cases}
		X[2m-1] &= E_{x,m} \\
		X[2m] &= E_{y,m}
	\end{cases}, \quad m = 1,2,\dots,N_v .
\end{equation}
The remaining subcarriers follow the Hermitian symmetry constraint	$X[N-k] = X^*[k], \quad k = 1,2,\dots,N/2-1$ .

After subcarrier mapping, the frequency-domain vector $\mathbf{X}$ is converted to the time-domain signal $\mathbf{x} = [x[0], x[1], \dots, x[N-1]]^T$ via the inverse fast Fourier transform (IFFT)	$\mathbf{x} = \mathbf{F}^H \mathbf{X}$, where $\mathbf{F}$ denotes the $N \times N$ normalized discrete Fourier transform (DFT) matrix with entries $\mathbf{F}_{p,q} = \frac{1}{\sqrt{N}} \exp(-j \frac{2\pi p q}{N})$.

To mitigate inter-symbol interference (ISI) caused by the channel delay spread, a cyclic prefix (CP) of length $N_{cp}$ is added to $\mathbf{x}$. Following parallel-to-serial (P/S) conversion and digital-to-analog (DAC) conversion, a direct current (DC) bias $B_{\text{DC}}$ is applied to the analog signal to guarantee its non-negativity, and the resulting signal drives the intensity modulation of the optical source (i.e., LED or laser diode).

The optical signal propagates through the channel and is detected by a photodetector (PD) at the receiver, which performs optical-to-electrical (O/E) conversion. After DC bias removal and analog-to-digital conversion (ADC), the discrete received sequence $\mathbf{y}$ is obtained after removing the CP.

Applying the fast Fourier transform (FFT) to $\mathbf{y}$ yields the received frequency-domain vector
\begin{equation} \label{eq:defY}
	\mathbf{Y} = \mathbf{F}\mathbf{y} = \mathbf{H}\mathbf{X} + \mathbf{Z},
\end{equation}
where $\mathbf{H}$ is an $N \times N$ diagonal matrix representing the frequency-domain channel response. The vector $\mathbf{Z} = [Z[0], Z[1], \dots, Z[N-1]]^T$ denotes additive white Gaussian noise (AWGN) in the frequency domain with $Z[k] \sim \mathcal{CN}(0,\sigma_z^2)$. The variance $\sigma_z^2$ corresponds to the continuous-time noise power spectral density $N_0$, such that $\Re\{Z[k]\}$ and $\Im\{Z[k]\}$ are independent with variance $N_0/2$.

The received adjacent subcarrier pairs are then fed into the VPM demapper to extract the 3D Stokes parameters for CSI-free detection. The details of the VPM mapping and demapping are given in Sections II-B and II-C, respectively.

\subsection{VPM Mapping Strategy}
The fundamental concept of VPM is inspired by the physical polarization state of light, which is mathematically described by the Jones vector and Stokes parameters\cite{Henarejos_2018_Tcom}. However, rather than utilizing physical polarization multiplexing hardware, VPM constructs a ``virtual'' polarization state by pairing two adjacent subcarriers in an OFDM symbol. 

For notational simplicity in the subsequent derivations, we omit the block index $m$ and focus on a generic VPM block. Let $\mathbf{S} = [S_1, S_2, S_3]^T$ denote the Cartesian coordinate of a constellation point on the Poincaré sphere, as shown in Fig.~\ref{fig:vpm_constellation}(a). To design an $M$-VPM constellation, where $M = 2^{n_b}$ and $n_b$ is the number of bits denoted by the symbol, the $M$ ideal points must be distributed as uniformly as possible on the spherical surface to maximize the minimum Euclidean distance in the Stokes space.  However, distributing points uniformly on a sphere is a well-known challenging problem, commonly referred to sphere packing\cite{sloane1984packing}. The problem was originally posed by Tammes \cite{Henarejos_2018_Tcom} and does not reach a closed-form solution for arbitrary $M$. Optimal solutions have been reported for several small constellation sizes\cite{sloane2018tables}, less than 130 and suitable for $n_b \le 7$. To obtain the  constellations  for $M>128$, we adopt the spherical Fibonacci lattice method \cite{Fan_2025_TCom}. This approach can generate a nearly uniform point distribution for any given $M$. Specifically, the integer index $i\in\{0,1,\dots,M-1\}$ is mapped to the angular pair $(\theta, \phi)$ in spherical coordinate as $\theta_i = \arccos\left(1 - \frac{2i + 1}{M}\right) $ and $	\phi_i = \left( \frac{2\pi i}{\Phi} \right) \bmod 2\pi $, where $\Phi = (1 + \sqrt{5})/2$ is the golden ratio\cite{Keinert_2015_ACM}. Then, the input data maps to the corresponding angular pair $(\theta, \phi)$, which defines the 3D Stokes parameters:
\begin{align}
	S_1 &= E_s \cos \theta \nonumber \\
	S_2 &= E_s \sin \theta \cos \phi \label{eq:spherical_def} \\
	S_3 &= E_s \sin \theta \sin \phi \nonumber
\end{align}
where $S_0 = E_s$ denotes the power of a VPM block. The spherical constraint $S_1^2 + S_2^2 + S_3^2 = S_0^2 = E_s^2$ is also satisfied. The  geometric relationship between the Stokes parameters and the angular pair $(\theta, \phi)$ is depicted in Fig.~\ref{fig:vpm_constellation}(b).

\begin{figure}[!t]
	\centering
	\includegraphics[width=0.7\columnwidth]{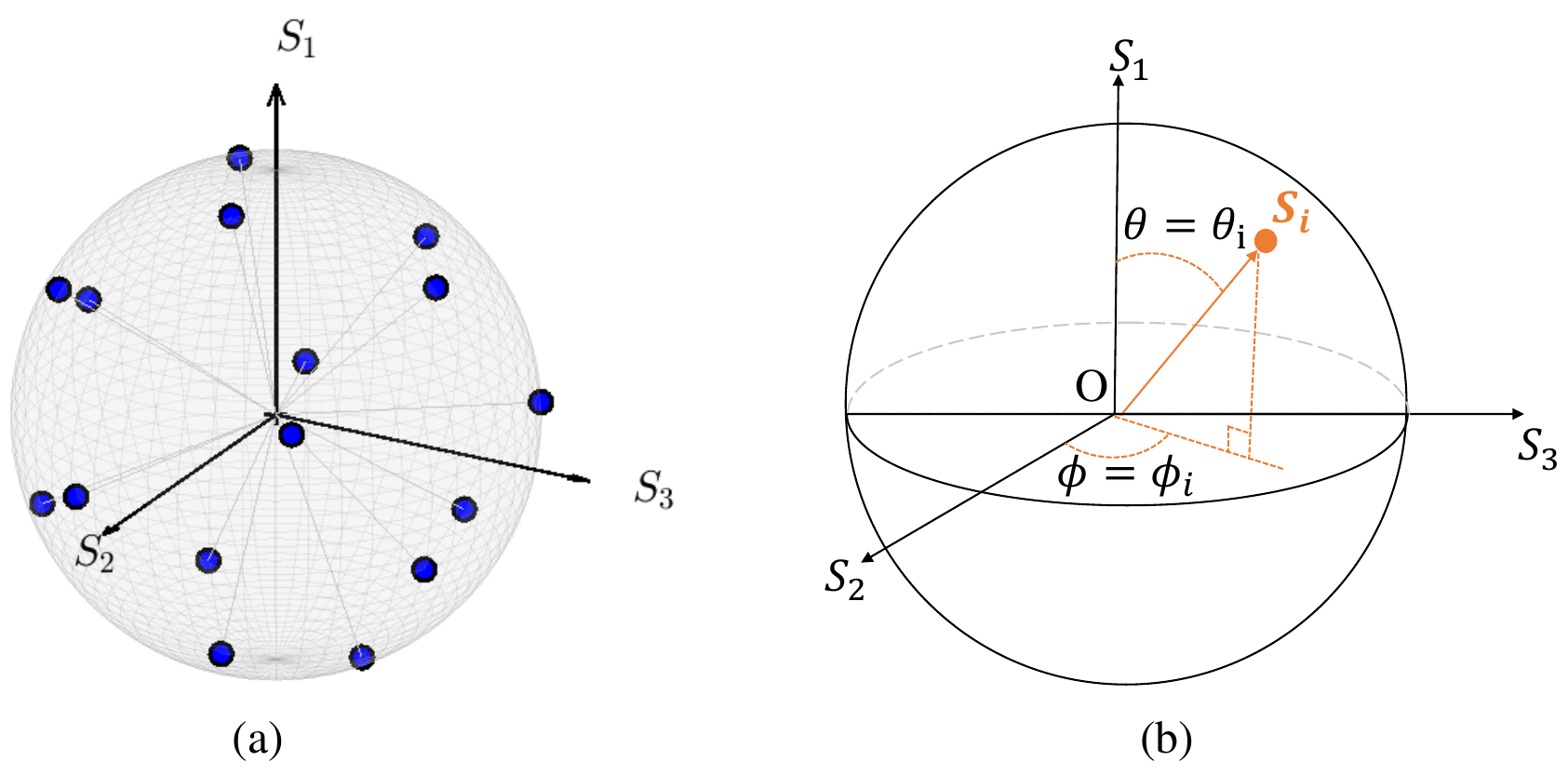}
	\caption{Geometric representation of the VPM mapping strategy in the 3D Stokes space. (a) Example of a 16-VPM constellation. (b) Geometric definition of the spherical coordinates $(\theta, \phi)$ corresponding to the Stokes parameters $(S_1, S_2, S_3)$ for the constellation point $\mathbf{S}_i$.}
	\label{fig:vpm_constellation}
\end{figure}

Once the Stokes vector $\mathbf{S}$ (or equivalently, the angular pair $\theta$ and $\phi$) is determined, it is then transformed into a Jones vector $\mathbf{E} = [E_x, E_y]^T$, where $E_x$ and $E_y$ act as the complex-valued symbols allocated to the $k$-th and $(k+1)$-th OFDM subcarriers. The transformation from the Jones vector to the Stokes parameters is defined as\cite{Henarejos_2018_Tcom}
\begin{align}
	S_0 &= |E_x|^2 + |E_y|^2 = E_s \nonumber \\
	S_1 &= |E_x|^2 - |E_y|^2 \label{eq:stokes_def} \\
	S_2 &= 2 \Re\{E_x E_y^*\} \nonumber \\
	S_3 &= -2 \Im\{E_x E_y^*\} \nonumber
\end{align}

It is noticed that \eqref{eq:stokes_def} maps the amplitudes of $E_x$ and $E_y$ to $S_0$ and $S_1$, while $S_2$ and $S_3$ depend on the relative phase difference between those two symbols. Finally, by substituting \eqref{eq:spherical_def} into \eqref{eq:stokes_def} and introducing an arbitrary initial common phase $\phi_c$, the transmitted VPM symbol block can be expressed as
\begin{align}\label{eq:vpm_generation}
	E_x &= \sqrt{E_s} \cos\left(\frac{\theta}{2}\right) \exp\left({j\phi_c}\right) \nonumber \\
	E_y &= \sqrt{E_s} \sin\left(\frac{\theta}{2}\right) \exp\left[{j(\phi_c + \phi)}\right]
\end{align}

The introduction of $\phi_c$ provides a degree of freedom. By randomly generating $\phi_c \in [-\pi, \pi)$ for each VPM block, the system can realize the selected mapping (SLM) technique\cite{cho2010mimo}, reducing the peak-to-average power ratio (PAPR) of the OFDM signal. More importantly, since the Stokes receiver exclusively evaluates the relative phase $\phi$, this PAPR reduction can be accomplished without any side information. The generated VPM block $[E_x, E_y]^T$ is subsequently fed into the IFFT module.

\subsection{CSI-free VPM Demapping in Stokes Space}
For conventional modulation formats such as QAM and PSK, absolute phase recovery is typically required and relies on accurate CE and equalization. In contrast, the proposed VPM scheme dispenses with this requirement by exploiting the self-cancellation property of its block structure in the Stokes space.

Assuming the optical channel response is approximately flat over the narrow band containing the paired adjacent subcarriers, the channel coefficients can be considered as identical, i.e., $H_k \approx H_{k+1} \approx H$.  $Y_x=H E_x + Z_x $ and $Y_y= H E_y + Z_y$ denote the received frequency-domain signals of a VPM block, where $Z_x$ and $Z_y$ are independent and identically distributed AWGN.

At the receiver, the VPM demapper reconstructs the Stokes parameters $\hat{\mathbf{S}} = [\hat{S}_1,\hat{S}_2,\hat{S}_3]^T$ according to
\begin{align}
	\hat{S}_1 &= |Y_x|^2 - |Y_y|^2 \nonumber\\
	\hat{S}_2 &= 2\Re\{Y_xY_y^*\} \label{eq:rx_stokes}\\
	\hat{S}_3 &= -2\Im\{Y_xY_y^*\}. \nonumber
\end{align}

Then, $Y_x Y_y^*$  can be expressed as	$Y_xY_y^*=(HE_x)(HE_y)^*=|H|^2(E_xE_y^*)$. Hence, the absolute channel phase $\angle H$ is canceled by the conjugate operation. The reconstructed Stokes vector is therefore scaled only by the real channel power gain $|H|^2$, leading to $	\hat{\mathbf{S}}\approx |H|^2\mathbf{S}+\Delta\mathbf{S}$, where $\Delta\mathbf{S}$ represents the noise vector in the Stokes space.

Since the scaling factor $|H|^2$ is a positive real scalar, it only changes the magnitude of the received vector without affecting its spatial orientation in the Stokes space. Therefore, symbol detection can be formulated as finding the ideal constellation point $\mathbf{S}_i$ on the Poincaré sphere that maximizes the normalized spatial correlation with the received vector $\hat{\mathbf{S}}$
\begin{equation} \label{eq:detection_metric}
	\hat{i} = \arg\max_{i \in \{0, \dots, M-1\}} \frac{\hat{\mathbf{S}}^T \mathbf{S}_i}{\|\hat{\mathbf{S}}\| \cdot \|\mathbf{S}_i\|},
\end{equation}
which is equivalent to the minimum angular distance criterion.

To illustrate the CSI-free property, we evaluate the denominator of \eqref{eq:detection_metric} in the noise-free scenario. The Euclidean norm of the received Stokes vector becomes $\|\hat{\mathbf{S}}\| = \||H|^2 \mathbf{S}\| = |H|^2 \|\mathbf{S}\|$. Consequently, the scalar factor $|H|^2$ appears in both the numerator and denominator as$
	\frac{(|H|^2 \mathbf{S})^T \mathbf{S}_i}{|H|^2 \|\mathbf{S}\| \cdot \|\mathbf{S}_i\|}
	=
	\frac{\mathbf{S}^T \mathbf{S}_i}{\|\mathbf{S}\| \cdot \|\mathbf{S}_i\|}
$

This result shows that the unknown channel gain $|H|^2$ is canceled during the normalization process. Hence, CE is not required for VPM-based optical OFDM, enabling a low-latency solution for optical transmission.

\section{SER Analysis of VPM-based OFDM over AWGN Channels}
In this section, the  SER performance of the proposed VPM scheme is analytically evaluated over an AWGN channel, assuming a flat channel with unit gain. Without loss of generality, we consider a generic transmitted VPM block $\mathbf{E}_i$ that belongs to the $M$-VPM constellation, where $i \in \{0, \dots, M-1\}$, to facilitate the SER derivation. According to \eqref{eq:defY}, the received Jones vector corresponding to a subcarrier pair can be expressed as
\begin{equation}
	\mathbf{Y} = \mathbf{E}_i + \mathbf{Z},
\end{equation}
The average symbol signal-to-noise ratio (SNR) is defined as ${\gamma}_s = E_s/N_0$. Therefore, the received vector $\mathbf{Y}$ follows a joint complex Gaussian distribution $f_\mathbf{Y}(\mathbf{Y})$, which can be written as
\begin{equation} \label{eq:complexGaussDistri}
	f_\mathbf{Y}(\mathbf{Y}) = \frac{1}{\pi^2 N_0^2} \exp \left( -\frac{\|\mathbf{Y} - \mathbf{E}_i\|^2}{N_0}\right)
\end{equation}

\subsection{SER Derivation and Closed-Form Expression}
In accordance with the VPM symbol detection rule in \eqref{eq:detection_metric}, the angular distribution of the received signal relative to the transmitted symbol in the Stokes space must be characterized. This requires mapping the received VPM block from the 2D complex Jones domain to the 3D Stokes space. Therefore, the geometric relationship between the two representations needs to be established, as given in the following lemma.
\begin{lemma}\label{lemma:inner_product}
	Let $\mathbf{E}_i,\mathbf{Y}\in\mathbb{C}^2$ denote the transmitted and received Jones vectors, and let $\mathbf{S}_i,\hat{\mathbf{S}}\in\mathbb{R}^3$ be their corresponding unnormalized Stokes vectors with total powers $E_s=\|\mathbf{E}_i\|^2$ and $\hat{S}_0=\|\mathbf{Y}\|^2$, respectively. Then the squared magnitude of their complex inner product can be expressed as
	\begin{equation}\label{eq:lemma1}
		|\mathbf{Y}^H\mathbf{E}_i|^2
		=
		\frac{1}{2}\hat{S}_0E_s(1+\hat{\mathbf{s}}^T\mathbf{s}_i)
		=
		\frac{1}{2}\hat{S}_0E_s(1+\cos\beta)
	\end{equation}
	where $\mathbf{s}_i=\mathbf{S}_i/E_s$ and $\hat{\mathbf{s}}=\hat{\mathbf{S}}/\hat{S}_0$ are the normalized Stokes vectors, and $\beta$ denotes the angle between them, as shown in Fig. \ref{fig:rec_constellation} (a).  Furthermore, it follows that $\beta=2\vartheta$, where $\vartheta$ is the angle between the Jones vectors $\mathbf{E}_i$ and $\mathbf{Y}$.
\end{lemma}
\textit{Proof:} See Appendix~\ref{apdx:1}. \hfill $\blacksquare$

\begin{figure}[t]
	\centering
	\includegraphics[width=0.7\columnwidth]{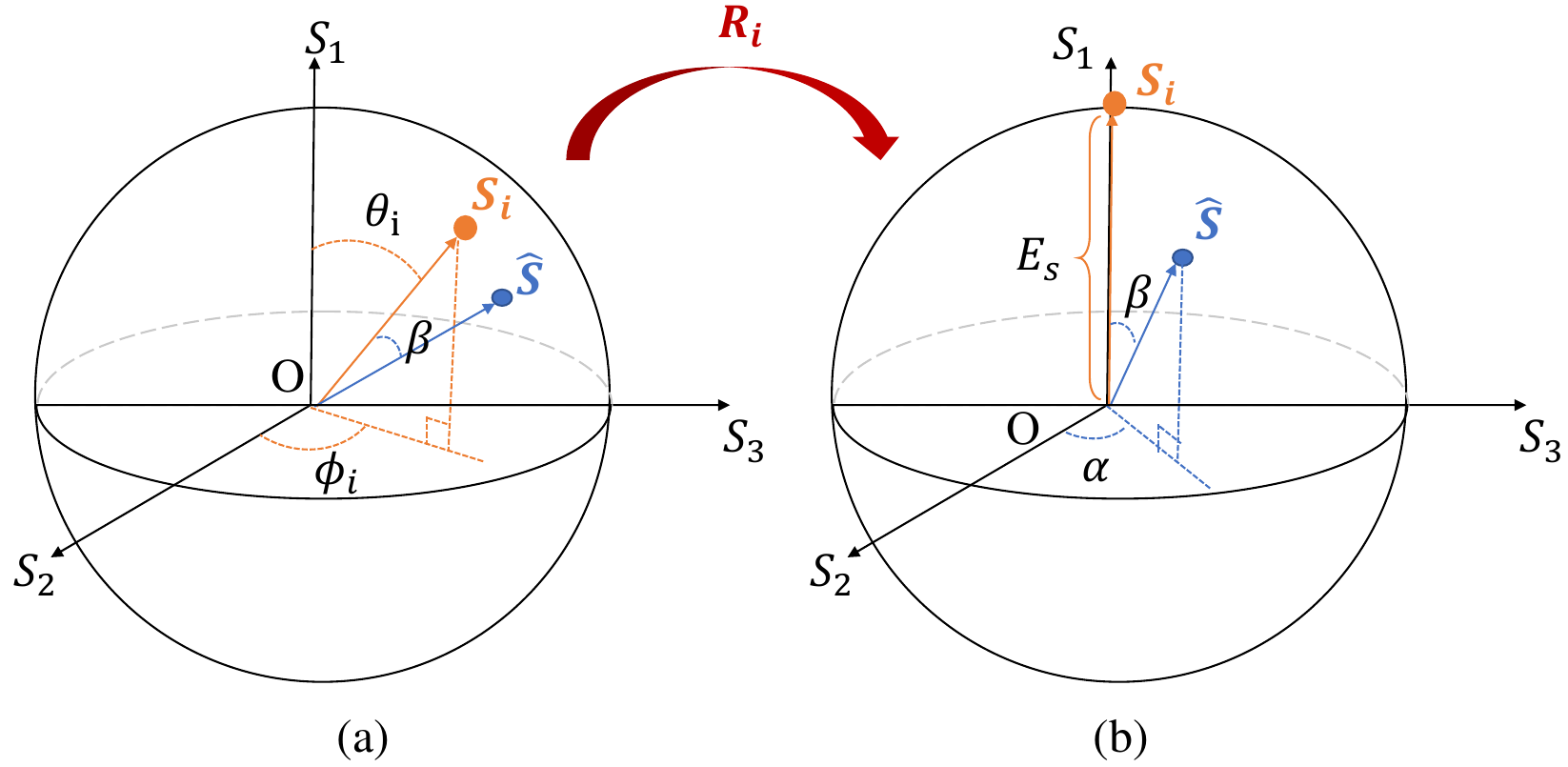}
	\caption{Geometric illustration of (a) the received vector $\hat{\mathbf{S}}$ and the transmitted Stokes symbol $\mathbf{S}_i$, and (b) the rotation that aligns $\mathbf{S}_i$ with the north pole of the Poincaré sphere.}
	\label{fig:rec_constellation}
\end{figure}

To facilitate the analysis of the correct detection probability when $\mathbf{S}_i$ is transmitted, we introduce a rotation technique. Specifically, we use a rotation matrix $\mathbf{R}_i$ that aligns the transmitted vector $\mathbf{S}_i$ with the reference north pole of the Poincaré sphere, namely, mapping $\mathbf{S}_i$ to $[E_s, 0, 0]^T$, as shown in Fig. \ref{fig:rec_constellation} (b). Therefore, in the rotated Poincaré sphere, $\beta$ corresponds to the zenith angle of $\hat{\mathbf{S}}$, while the azimuth angle is denoted by $\alpha \in [0,2\pi)$. Based on the spherical coordinates $(\theta_i, \phi_i)$ of the constellation point $\mathbf{S}_i$, $\mathbf{R}_i$ can be expressed as
\begin{equation} \label{eq:rotation_matrix}
	\mathbf{R}_i = \begin{bmatrix}
		\cos\theta_i & \sin\theta_i\cos\phi_i & \sin\theta_i\sin\phi_i \\
		-\sin\theta_i & \cos\theta_i\cos\phi_i & \cos\theta_i\sin\phi_i \\
		0 & -\sin\phi_i & \cos\phi_i
	\end{bmatrix}
\end{equation}
Since the rotation is orthogonal, it preserves Euclidean distances and angular relationships. Then we can achieve the following lemma.

\begin{lemma}\label{lemma:exact_pdf}
	After the rotation and given the average SNR $\gamma_s$, the  joint PDF of   $\beta \in [0, \pi]$ and $\alpha \in [0,2\pi)$ is derived as
		\begin{align}
	f_{\beta, \alpha}(\beta, \alpha | \gamma_s) & = \frac{1}{4\pi} \exp\left[{-\frac{\gamma_s}{2}(1 - \cos\beta)}\right]  \nonumber \\
	& \times \left[ 1 + \frac{\gamma_s}{2}(1 + \cos\beta) \right] \sin(\beta) \label{eq:exact_pdf_beta}.
\end{align}

\end{lemma}
\textit{Proof:} See Appendix \ref{apdx:2}. \hfill $\blacksquare$

The SER of the VPM symbol on each subcarrier pair in an OFDM symbol can be obtained by integrating the joint angular PDF in \eqref{eq:exact_pdf_beta} within the correct decision region of the corresponding constellation point. Finally, based on Lemma~\ref{lemma:exact_pdf}, the following theorem provides the closed-form expression for the average SER over an OFDM symbol with an arbitrarily designed $M$-VPM constellation.
\begin{theorem}\label{thm:exact_ser}
	Assuming that the $M$ constellation symbols of the $M$-VPM scheme are equiprobable, the closed-form expression of the average SER for the VPM-based OFDM system over the AWGN channel can be derived using Gauss–Legendre quadrature \cite{1969_Froberg_book} as
	\begin{align} 
	\bar{P}_e^{AWGN} (\gamma_s) &= \frac{1}{N_v}\sum_{m=1}^{N_v}\left(1-\frac{1}{M}\sum_{i=0}^{M-1}P_c^{(i)}(\gamma_s)\right) \nonumber\\
	&\approx \frac{1}{2M}\sum_{i=0}^{M-1}\sum_{q=1}^{N_Q}w_q  \frac{1+\cos\left[\beta_{\max}^{(i)}(\pi x_q+\pi)\right]}{2} \nonumber\\ 
	& \times \exp\left(-\frac{\gamma_s}{2}\left(1-\cos\left[\beta_{\max}^{(i)}(\pi x_q+\pi)\right]\right)\right),
\end{align}
	where $P_c^{(i)}(\gamma_s)$ denotes the probability of correct detection when the $i$-th constellation symbol is transmitted, which can be expressed as
	\begin{equation}
		P_c^{(i)}(\gamma_s) = \int_{0}^{2\pi}\int_{0}^{\beta_{\max}^{(i)}(\alpha)}f_{\beta, \alpha}(\beta, \alpha | \gamma_s)\diff \alpha \diff \beta.  \label{eq:Pc}
	\end{equation}
	$\beta_{\max}^{(i)}(\alpha)$ denotes the maximum zenith angle within which the $i$-th constellation point can be correctly detected for a given azimuth angle $\alpha$, and it is given by
		\begin{equation} \label{eq:beta_max}
		\beta_{\max}^{(i)}(\alpha) = \mathrm{arccot}\left(\max_{\substack{k \in \{0, \dots, M-1\} \\ k \ne i}} \cot\left(\frac{\beta_k}{2}\right)\cos(\alpha - \alpha_k) \right),
	\end{equation}
	where $\beta_k$ and $\alpha_k$ denote the zenith and azimuth angles of the remaining $M-1$ constellation points after the rotation, respectively. $x_q$ denotes the $q$-th root of the Legendre polynomial ${P_{{N_Q}}}\left( x  \right) = \frac{1}{{{2^{{N_Q}}}{N_Q}!}}\frac{{{{\mathop{\rm d}\nolimits} ^{{N_Q}}}}}{{{\mathop{\rm d}\nolimits} {x ^{{N_Q}}}}}\left[ {{{\left( {{x ^2} - 1} \right)}^{{N_Q}}}} \right]$ of degree ${N_Q}$ within the standard interval $\left[ { - 1,1} \right]$; and the corresponding weight ${w_q}$ is determined as ${w_q} = \frac{2}{\left(1 - x_q^2\right) \left[ P'_{N_Q}\left(x_q\right) \right]^2}$.
\end{theorem}
\textit{Proof:} See Appendix \ref{apdx:3}. \hfill $\blacksquare$

\subsection{SER Bounds and High-SNR Approximation}
Deriving  upper and lower bounds provides closed-form analytical insights without requiring complex numerical integration over the exact boundaries. By replacing the dynamic integration limit $\beta_{\max}^{(i)}(\alpha)$ with a  constant boundary $\beta_c$, the integral of \eqref{eq:Pc} over the azimuth $\alpha \in [0, 2\pi]$ evaluates to $2\pi$, canceling the $1/(2\pi)$ coefficient in \eqref{eq:app_C_inner_integral}. Hence, the error probability simplifies to a closed-form expression
\begin{equation} \label{eq:ser_constant_bound}
	P_e(\beta_c) = \frac{1+\cos\beta_c}{2} \exp\left(-\frac{\gamma_s}{2}(1-\cos\beta_c)\right).
\end{equation}
Based on \eqref{eq:ser_constant_bound}, the SER bounds are formulated by substituting specific geometric constants.

\subsubsection{Upper Bound}
 By setting the constant boundary to $\cos\beta_{\text{UB}} = \sqrt{1 - d_{min}^2/4}$, where $d_{min}$ is the minimum  Euclidean distance in the Poincaré sphere, the integration region used to calculate the correct probability $P_c$ is  smaller than the exact region. Consequently, this geometry underestimates $P_c$, thereby establishing a upper bound for  SER:
\begin{align} \label{eq:ser_ub}
	\bar{P}_e^{U} (\gamma_s) = \frac{1+\sqrt{1 - \frac{d_{min}^2}{4}}}{2}  \exp\left[-\frac{\gamma_s}{2}\left(1-\sqrt{1 - \frac{d_{min}^2}{4}}\right)\right].
\end{align}

\subsubsection{Lower Bound}
Conversely, a lower bound is obtained by replacing the integration region with an area-equivalent spherical cap, yielding the constant boundary $\cos\beta_{\text{LB}} = 1 - 2/M$. It  evaluates a larger correct probability than the exact irregular polygonal boundary, thus providing a  lower bound for  SER:
\begin{equation} \label{eq:ser_lb}
	\bar{P}_e^{L} (\gamma_s) = \left(1-\frac{1}{M}\right)\exp\left(-\frac{\gamma_s}{M}\right).
\end{equation}

The high-SNR  approximation of SER is provided in the following theorem.

\begin{theorem}\label{thm:high_snr}
	In the high-SNR regime ($\gamma_s \to \infty$), the equivalent noise in the Stokes space asymptotically approaches a 3D isotropic Gaussian distribution. Consequently, the average SER of the VPM scheme can be approximated by the following closed-form  expression in the high-SNR regime:
	\begin{equation} \label{eq:ser_high_snr}
		\bar{P}_e^{\infty}  \approx \bar{N}_{min} \times Q\left( \sqrt{\gamma_s \left(1 - \cos\beta_{min}\right)} \right),
	\end{equation}
	where $Q(\cdot)$ is the Gaussian Q-function\cite{proakis2008digital}, $\bar{N}_{min}$ is the average number of nearest neighbors over the entire constellation. For an equiprobable $M$-VPM constellation, it is defined as
	\begin{equation} \label{eq:N_min_avg}
		\bar{N}_{min} = \frac{1}{M} \sum_{i=1}^{M} N_i,
	\end{equation}
	where $N_i$ denoting the number of adjacent Stokes symbols at the minimum angular distance $\beta_{min}$ from the $i$-th symbol. Mathematically, $N_i$ is the cardinality of the nearest-neighbor set:
	\begin{equation} \label{eq:N_i_def}
		N_i = \left| \left\{ k \in \{1, \dots, M\} \setminus \{i\} \mid \mathbf{s}_{i}^T \mathbf{s}_{k} = \cos\beta_{min} \right\} \right|.
	\end{equation}
\end{theorem}
\textit{Proof:} See Appendix \ref{apdx:4}. \hfill $\blacksquare$

\textit{Remark:} The high-SNR approximation fundamentally relies on the assumption of identical, symmetric pairwise error probabilities among nearest neighbors. Consequently, \eqref{eq:ser_high_snr} achieves the highest accuracy for topologies where the constellation points are uniformly distributed across the Poincaré sphere.
\section{SER analysis of VPM-Based OFDM over Optical Wireless Channels}
In practical OWC systems, the physical channel exhibits variations across both frequency and time. In the frequency domain, hardware impairments (e.g., the limited modulation bandwidth of LEDs) and multipath propagation induce frequency-selective effects, resulting in frequency-dependent amplitude attenuation and phase rotation across adjacent subcarriers. In the time domain, atmospheric turbulence induces large-scale random irradiance fluctuations, resulting in time-varying stochastic fading. 

A primary advantage of the proposed VPM-based OFDM scheme is its CSI-free detection capability, which eliminates the need for CE and  reduces system latency. However, operating without active channel equalization implies that the receivers operate under uncompensated physical impairments. In order to  quantify the impact of practical OWC channel impairments on VPM-based OFDM scheme, this section extends the analysis to evaluate its performance under the joint influence of frequency-selective channel variations and time-varying atmospheric fading.

\subsection{Angular Drift and Average SER under Frequency-Selective Channels}
Consider an OFDM subcarrier pair $(k, k+1)$ carrying the transmitted Jones vector $\mathbf{E}_i = [E_{x}, E_{y}]^T$. Let $A = |E_{x}|^2$ and $B = |E_{y}|^2$ denote the power allocated to the two subcarriers, such that the total symbol power is $E_s = A + B$.

Neglecting additive noise, the frequency-domain channel response introduces subcarrier-dependent gains $H_k$ and $H_{k+1}$. Define the differential gain as $\Delta H_k = H_{k+1} - H_k$. The resulting noise-free received Jones vector is
\begin{equation}\label{eq:received_jones}
	\tilde{\mathbf{Y}} =
	\begin{bmatrix}
		H_k E_x \\
		H_{k+1} E_y
	\end{bmatrix}
	= H_k \mathbf{E}_i + \mathbf{e},
\end{equation}
where $\mathbf{e} = [0, \Delta H_k E_y]^T$ captures the distortion caused by the frequency-selective channel.

Based on  Lemma~\ref{lemma:inner_product}, the channel-induced angular drift on the Poincaré sphere, quantified by the mismatch angle $\beta_{\text{mis}}$ between the distorted Stokes vector and the target point, is derived by expanding $|\tilde{\mathbf{Y}}^H \mathbf{E}_i|^2$  as $	\sin^2\left(\frac{\beta_{\text{mis}}}{2}\right) = \frac{A B |\Delta H_k|^2}{E_s S_0^{\text{out}}},
$
where $S_0^{\text{out}} = |H_k|^2 A + |H_k + \Delta H_k|^2 B$ denotes the instantaneous received power. For closely spaced subcarriers, that is, $\Delta H_k \ll H_k$ and $S_0^{\text{out}} \approx |H_k|^2 A + |H_k |^2 B =  |H_k|^2E_s$, we have
\begin{equation}\label{eq:beta_miss_^2}
	\beta_{\text{mis}}^2 \approx 4ab (|\Delta H_k|^2 / |H_k|^2),
\end{equation} 
where $a = A/E_s$ and $b = B/E_s$ denote the power allocation ratios for the subcarrier pair.

In practical IM/DD OWC systems, the equivalent baseband frequency response $H_{\text{sys}}(f)$ is determined by the cascade of the LED modulation bandwidth and the multipath temporal dispersion \cite{Uysal_CM_2017}. The LED is commonly modeled as a first-order low-pass filter with a 3-dB cutoff frequency $f_c$, yielding the frequency response $H_{\text{LED}}(f) = \frac{1}{1 + j(f/f_c)}$ \cite{Ghassemlooy_book}. The indoor multipath channel is often characterized by an exponential-decay impulse response $h_{\text{mp}}(t) = \frac{1}{\tau_{\text{rms}}} e^{-t/\tau_{\text{rms}}} u(t)$ with an RMS delay spread $\tau_{\text{rms}}$ \cite{Jungnickel_2002_JSAC}, where $u(t)$ denotes the unit step function. Its corresponding frequency response is $H_{\text{mp}}(f) = \frac{1}{1 + j 2\pi f \tau_{\text{rms}}}$. Hence, the OWC system frequency response can be written as
\begin{equation}  \label{eq:H_sys}
	H_{\text{sys}}(f) = H_{\text{LED}}(f) H_{\text{mp}}(f).
\end{equation}

For a small subcarrier spacing $\Delta f$, the channel difference between adjacent subcarriers can be approximated using the first-order derivative
\begin{equation} \label{eq:Delta_H}
	\Delta H_{\text{sys}} \approx \frac{\diff H_{\text{sys}}(f)}{\diff f} \big|_{f=f_k} \Delta f.
\end{equation}
Applying the product rule to \eqref{eq:H_sys} yields
\begin{equation}
	\frac{\diff H_{\text{sys}}(f)}{\diff f} = \frac{\diff H_{\text{LED}}(f)}{\diff f} H_{\text{mp}}(f) + H_{\text{LED}}(f) \frac{\diff  H_{\text{mp}}(f)}{\diff f}.
\end{equation}
Dividing both sides of \eqref{eq:Delta_H} by $H_{\text{sys}}(f) = H_{\text{LED}}(f) H_{\text{mp}}(f)$ provides the relative channel variation at $f_k$. To simplify the notation, we define the relative distortion rate as $\Lambda(f) = \frac{1}{H(f)} \frac{d H(f)}{df}$. The relative variation of the OWC system can then be expressed as the linear sum of the individual rates
\begin{equation}\label{eq:relative_distortion}
	\frac{\Delta H_{\text{sys}}}{H_{\text{sys}}(f_k)} \approx \left[ \Lambda_{\text{LED}}(f_k) + \Lambda_{\text{mp}}(f_k) \right] \Delta f,
\end{equation}
where the relative distortion rates for the LED and the multipath channel can be expressed as $
	\Lambda_{\text{LED}}(f_k) = \frac{-j / f_c}{1 + j(f_k/f_c)}
$
and
$
	\Lambda_{\text{mp}}(f_k) = \frac{-j 2\pi \tau_{\text{rms}}}{1 + j 2\pi f_k \tau_{\text{rms}}},
$
respectively. Substituting \eqref{eq:relative_distortion} into \eqref{eq:beta_miss_^2}, the squared spatial mismatch angle can be derived as
\begin{equation}\label{eq:beta_miss_final}
	\beta_{\text{mis}}^2 \approx 4ab (\Delta f)^2 \left| \frac{1/f_c}{1 + j(f_k/f_c)} + \frac{2\pi \tau_{\text{rms}}}{1 + j 2\pi f_k \tau_{\text{rms}}} \right|^2.
\end{equation}
This result indicates that the angular drift in the Stokes space is jointly determined by the LED cutoff frequency and the channel delay spread.

When additive complex Gaussian noise $\mathbf{Z}$ is considered, the random deflection on the Poincaré sphere is primarily determined by the noise component orthogonal to the noise-free signal $\tilde{\mathbf{Y}}$. By introducing the orthogonal projection matrix $\mathbf{P}_{\tilde{Y}}^{\perp} = \mathbf{I}_2 - \tilde{\mathbf{Y}} \tilde{\mathbf{Y}}^H / \|\tilde{\mathbf{Y}}\|^2$, the projected orthogonal noise is given by $\mathbf{Z}^{\perp} = \mathbf{P}_{\tilde{Y}}^{\perp} \mathbf{Z}$. In the high-SNR regime, the small angular perturbation $\vartheta$ in the Jones space is closely approximated by the ratio of the orthogonal noise magnitude to the signal magnitude, $\vartheta \approx \|\mathbf{Z}^{\perp}\|/\|\tilde{\mathbf{Y}}\|$. Based on   $\mathbb{E}[\|\mathbf{Z}^{\perp}\|^2] = \mathbb{E}[\text{Tr}(\mathbf{Z}^H \mathbf{P}_{\tilde{Y}}^{\perp} \mathbf{Z})] = \text{Tr}(\mathbf{P}_{\tilde{Y}}^{\perp} \mathbb{E}[\mathbf{Z}\mathbf{Z}^H]) = N_0$ and the Lemma~\ref{lemma:inner_product}, the variance of the noise-induced angular drift is given by $\mathbb{E}[\beta_{\text{noise}}^2] = 4 \mathbb{E}[\vartheta^2] \approx 4N_0/\|\tilde{\mathbf{Y}}\|^2 = 4/\gamma_{\text{inst}}$, where  $\gamma_{\text{inst}} = \|\tilde{\mathbf{Y}}\|^2/N_0$ is the  instantaneous effective SNR.

On the curved Poincaré sphere, angular deflections do not add linearly. Under the high-SNR assumption, however, the local neighborhood around the target point can be approximated by a Euclidean tangent plane. In this tangent space, the total displacement vector is the vector sum of the deterministic mismatch and the random noise. Consequently, the mean-square angular drift in the frequency-selective OWC channels can be approximated by the sum of their squared magnitudes as follows
\begin{equation}\label{eq:beta_total_turb}
	\mathbb{E}[\beta_{\text{total}}^2 ] \approx \beta_{\text{mis}}^2 + \mathbb{E}[\beta_{\text{noise}}^2] = \beta_{\text{mis}}^2 +   \frac{4}{\gamma_{\text{inst}}}.
\end{equation}
From \eqref{eq:beta_total_turb}, as the effective SNR approaches infinity ($\gamma_{\text{inst}} \to \infty$), the total angular drift is asymptotically dominated by the channel-induced term $\beta_{\text{mis}}$. If this drift exceeds the maximum angular decision threshold (i.e., $\beta_{\text{mis}} > \max_{\alpha} \beta_{\max}^{(i)}(\alpha)$), the received Stokes vector crosses the decision boundary even in the absence of noise, resulting in an error floor in the high-SNR regime.

Next, we derive the average SER of the VPM-based OFDM  under OWC frequency-selective channels. For the $m$-th VPM block carrying the $i$-th constellation symbol, the frequency-selective channel scales the effective SNR to $\gamma_{\text{eff},m}^{(i)} = \gamma_s \|\tilde{\mathbf{Y}}_{m,i}\|^2 / E_s$. With additive AWGN added to the noise-free signal $\tilde{\mathbf{Y}}_{m,i}$, the center of the conditional Stokes PDF shifts from the ideal normalized coordinate $\mathbf{s}_{i}$ to $\mathbf{s}'_{m,i}$. In order to calculate the error probability, a rotation matrix $\mathbf{R}'_{m,i}$ is also used to align the drifted center $\mathbf{s}'_{m,i}$ with the north pole of the Poincaré sphere and can be obtained through similar definition in \eqref{eq:rotation_matrix}. The corresponding decision hyperplanes are rotated accordingly, yielding the updated normal vectors $\mathbf{w}'_{m,i,j} = \mathbf{R}'_{m,i} (\mathbf{s}_{i} - \mathbf{s}_{j})$, where $j \neq i$ indexes the neighboring constellation points. 

\begin{theorem}\label{thm:exact_mismatch_ser}
	For a $M$-VPM-based OFDM system with $N_v$ active subcarrier blocks, the closed-form expression for average SER  under OWC frequency-selective channels is derived as
		\begin{equation}\label{eq:SER_mismatch_exact}
		\begin{split}
			\bar P_e^f (\gamma_s) &= \frac{1}{N_v M} \sum_{m=1}^{N_v} \sum_{i=0}^{M-1} \left( 1 - P_{c,m}^{(i)} \right) \\
			&\approx \frac{1}{2N_v M} \sum_{m=1}^{N_v} \sum_{i=0}^{M-1}  \sum_{q=1}^{N_Q} w_q \frac{1+\cos\beta_m^{(i)}(\pi x_q + \pi)}{2} \\
			&\quad \times \exp\left(-\frac{\gamma_{\text{eff},m}^{(i)}}{2}\left(1-\cos\beta_m^{(i)}(\pi x_q + \pi)\right)\right).
		\end{split}
	\end{equation}
	where $P_{c,m}^{(i)}$ is the probability of correct detection for the $i$-th symbol on the $m$-th block and can be expressed as
	\begin{align} \label{eq:Pc_mismatch_exact}
		P_{c,m}^{(i)} &= 1 - \frac{1}{2\pi} \int_{0}^{2\pi} \frac{1+\cos\beta_m^{(i)}(\alpha)}{2} \nonumber \\
		&\times \exp\left(-\frac{\gamma_{\text{eff},m}^{(i)}}{2}\left(1-\cos\beta_m^{(i)}(\alpha)\right)\right) \diff \alpha.
	\end{align}
	Given the azimuth angle $\alpha$, the  integration boundary of the zenith angle is
	\begin{equation} \label{eq:beta_m_f}
		\beta_m^{(i)}(\alpha) = \text{arccot} \left( \max_{\substack{j \in \{0, \dots, M-1\} \\ j \ne i}} \left[ \frac{-w'_{j,x} \cos\alpha - w'_{j,y} \sin\alpha}{w'_{j,z}} \right] \right),
	\end{equation}
	 where $[w'_{j,x}, w'_{j,y}, w'_{j,z}]^T$ denoting the Cartesian components of the normal vector $\mathbf{w}'_{m,i,j}$.
\end{theorem}
\textit{Proof:} See Appendix E. \hfill $\blacksquare$

\subsection{Average SER under Atmospheric Turbulence Channels}
In practical OWC links, especially in outdoor applications, large-scale atmospheric turbulence induces random irradiance fluctuations. In this subsection, we extend Theorem~\ref{thm:exact_mismatch_ser} to account for atmospheric turbulence and evaluate the SER of the proposed VPM-based OFDM system. Since the coherence time of turbulence-induced fading is typically on the order of milliseconds\cite{Ghassemlooy_book}, the fading coefficient $h$ remains nearly constant over one OFDM symbol and acts as a slow flat-fading multiplier. Because the photoelectric current is proportional to the received optical irradiance, the electrical SNR scales with $h^2$. Therefore, for the $i$-th symbol on the $m$-th subcarrier block, the turbulence-impaired effective SNR becomes $\gamma(h) = h^2 \gamma_{\text{eff},m}^{(i)}$.

To compute the average probability of correct detection $\overline{P}_{c,m}^{(i)}$, the  probability in \eqref{eq:Pc_mismatch_exact} is averaged over the turbulence-induced fading $h$ with PDF $f_H(h)$. By interchanging the order of integration, we have
\begin{align}\label{eq:avg_Pc_swapped}
	\overline{P}_{c,m}^{(i)} &= 1 - \frac{1}{2\pi} \int_{0}^{2\pi} \frac{1+\cos\beta_m^{(i)}(\alpha)}{2} \nonumber \\
	&\quad \times \underbrace{ \int_{0}^{\infty} \exp\Big(-h^2 g_m^{(i)}(\alpha)\Big) f_H(h) \mathrm{d}h }_{\triangleq \mathcal{I}_{\text{fade}}\left(g_m^{(i)}(\alpha)\right)} \mathrm{d}\alpha,
\end{align}
where the angular-dependent exponential argument is defined as $g_m^{(i)}(\alpha) = \frac{\gamma_{\text{eff},m}^{(i)}}{2}\left(1-\cos\beta_m^{(i)}(\alpha)\right)$, and $\mathcal{I}_{\text{fade}}(\cdot)$ denotes the fading integral.

For LED-based outdoor links, the irradiance fluctuation $h$ is commonly modeled as a LN random variable \cite{Sharda_2022_TCom, Cao_2025_OE} with PDF  $
	f_{LN}(h) = \frac{1}{\sqrt{2\pi}\sigma_l h} \exp\Bigg[-\frac{(\ln h + \sigma_l^2/2)^2}{2\sigma_l^2}\Bigg],
$
where $\sigma_l^2$ denotes the log-intensity variance. Applying the variable change $x = (\ln h - \mu_l)/(\sqrt{2}\sigma_l)$, the fading integral $\mathcal{I}_{\text{fade}}$ is transformed into a standard Gaussian-weighted form. It can then be efficiently evaluated using Gauss–Hermite quadrature as \cite{1969_Froberg_book}
\begin{equation}\label{eq:Integral_LN}
	\begin{split}
		\mathcal{I}_{LN}\big(g_m^{(i)}(\alpha)\big) &\approx \frac{1}{\sqrt{\pi}} \sum_{p=1}^{N_G} \omega_p \exp\bigg( -\frac{\gamma_{\text{eff},m}^{(i)}}{2} \\
		&\times \big(1-\cos\beta_m^{(i)}(\alpha)\big) \exp\left({2\sqrt{2}\sigma_l x_p - \sigma_l^2}\right) \bigg),
	\end{split}
\end{equation}
where $x_p$ is the $p$-th root of the $N_G$-th order Hermite polynomial $H_{N_G}(x) = (-1)^{N_G} e^{x^2} \frac{\mathrm{d}^{N_G}}{\mathrm{d}x^{N_G}} e^{-x^2}$. The associated quadrature weights are $	\omega_p = \frac{2^{N_G + 1} N_G! \sqrt{\pi}}{\left[ H'_{N_G}(x_p) \right]^2}$.

For outdoor FSO laser links, the GG distribution offers an acceptable statistical model for irradiance fluctuations\cite{Ghassemlooy_book,andrews2005laser,cao2015average}. Its PDF is expressed as
$
	f_{GG}(h) = \frac{2 (a_g b_g)^{\frac{a_g+b_g}{2}}}{\Gamma(a_g)\Gamma(b_g)} \, h^{\frac{a_g+b_g}{2}-1} K_{a_g-b_g}\big( 2\sqrt{a_g b_g h} \big),
$
where $\Gamma(\cdot)$ denotes the Gamma function, and $K_{\nu}(\cdot)$ is the modified Bessel function of the second kind. The parameters $a_g$ and $b_g$ correspond to the effective numbers of large- and small-scale scattering eddies, respectively, and are linked to the physical FSO channel via the Rytov variance $\sigma_R^2$ as
\begin{align}
	a_g &= \left\{ \exp\left[ \frac{0.49 \sigma_R^2}{(1 + 1.11 \sigma_R^{12/5})^{7/6}} \right] - 1 \right\}^{-1}, \label{eq:ag_Rytov} \\
	b_g &= \left\{ \exp\left[ \frac{0.51 \sigma_R^2}{(1 + 0.69 \sigma_R^{12/5})^{5/6}} \right] - 1 \right\}^{-1}. \label{eq:bg_Rytov}
\end{align}
Using \cite{wolfram_meijerg_integral} and \cite[9.34.3]{TableofInte}, we can express $\exp(-h^2 g_m^{(i)}(\alpha))$ and $K_{\nu}(\cdot)$ in terms of Meijer G-functions, respectively. Then, based on \cite[2.24.1.1]{prudnikov1990integrals}, $\mathcal{I}_{GG}\big(g_m^{(i)}(\alpha)\big) $ can be derived in closed-form as follows
\begin{align}\label{eq:Integral_GG}
	&\mathcal{I}_{GG}\big(g_m^{(i)}(\alpha)\big) = \frac{2^{a_g+b_g-2}}{\pi \Gamma(a_g)\Gamma(b_g)} \nonumber \\ 
	& \quad \times G_{4,1}^{1,4} \left( \frac{16 \, g_m^{(i)}(\alpha)}{(a_g b_g)^2} \ \middle| \ \begin{array}{c} \frac{1-a_g}{2}, \frac{2-a_g}{2}, \frac{1-b_g}{2}, \frac{2-b_g}{2} \\ 0 \end{array} \right).
\end{align}

By substituting \eqref{eq:Integral_LN} or \eqref{eq:Integral_GG} back into \eqref{eq:avg_Pc_swapped}, and applying the Theorem \ref{thm:exact_mismatch_ser}, the closed-form expression for average SER  over OWC turbulent channels is derived as
\begin{equation}\label{eq:SER_global}
	\begin{split}
		\bar P_e^t (\gamma_s) &\approx \frac{1}{2N_v M} \sum_{m=1}^{N_v} \sum_{i=0}^{M-1} \sum_{q=1}^{N_Q} w_q \frac{1+\cos\beta_m^{(i)}(\pi x_q + \pi)}{2} \\
		&\quad \times \mathcal{I}_{\text{fade}}\left( \frac{\gamma_{\text{eff},m}^{(i)}}{2} \big(1-\cos\beta_m^{(i)}(\pi x_q + \pi)\big) \right).
	\end{split}
\end{equation}

This unified expression accounts for hardware impairments, multipath propagation, and atmospheric fading, providing a numerically tractable framework to evaluate the SER of the proposed VPM-based OFDM system for arbitrary constellation designs.

\section{Results and Discussions}
In this section, we present the theoretical SER results of the proposed VPM-based DCO-OFDM system and verify them using Monte Carlo (MC) simulations. Unless otherwise specified, the adopted parameters are  \cite{Dissanayake_2013_JLT,chen2017efficient,Ghassemlooy_book}: $N=1024$, $N_{cp}=256$, $f_c=10$ MHz, $\tau_{\text{rms}}=10$ ns, $N_G=30$, $N_Q=30$, and the total bandwidth $B_w=20$ MHz. A DC bias equal to three times the standard deviation of the OFDM signal is adopted, and the constellation points are uniformly distributed \cite{sloane2018tables}. Here, we adopt the average SNR per subcarrier ($\bar{\gamma}_{sub}$) to ensure a fair energy-efficiency comparison between VPM and QAM, where QAM occupies a single subcarrier whereas VPM spans a subcarrier pair.

\begin{figure}[htbp]
	\centering
	\includegraphics[width=0.65\columnwidth]{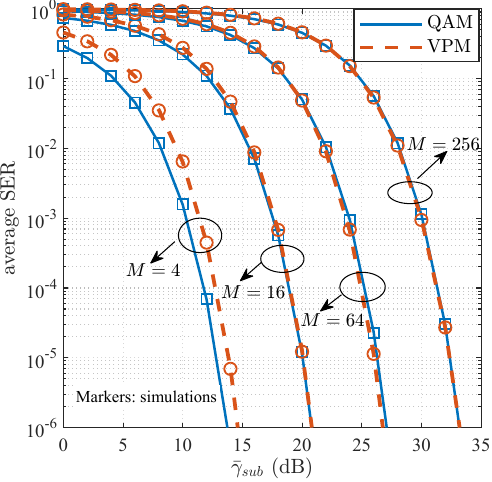}
	\caption{Average SER comparison between the proposed $M$-VPM and $M$-QAM over AWGN channels with $M \in \{4, 16, 64, 256\}$.}
	\label{fig:rslt1}
\end{figure}

Fig. \ref{fig:rslt1} compares the average SER of the proposed $M$-VPM and conventional $M$-QAM with DCO-OFDM over AWGN channels for $M \in \{4,16,64,256\}$. Theoretical results are shown as lines and MC simulations as markers. The SER curves of $M$-QAM are obtained by \cite[4.3-27]{proakis2008digital}. The simulation results closely match the theoretical curves, confirming the accuracy of the derived SER expressions. VPM and QAM exhibit very similar error performance. A small gap appears at $M=4$, but it quickly vanishes as $M$ increases. For $M=64$ and $M=256$, the SER curves almost overlap. For both modulation schemes, to achieve the same SER, the required $\bar{\gamma}_{sub}$ increases with $M$, since the minimum distance between constellation points decreases as more symbols are packed into the signal space.

\begin{figure}[htbp]
	\centering
	\includegraphics[width=0.8\columnwidth]{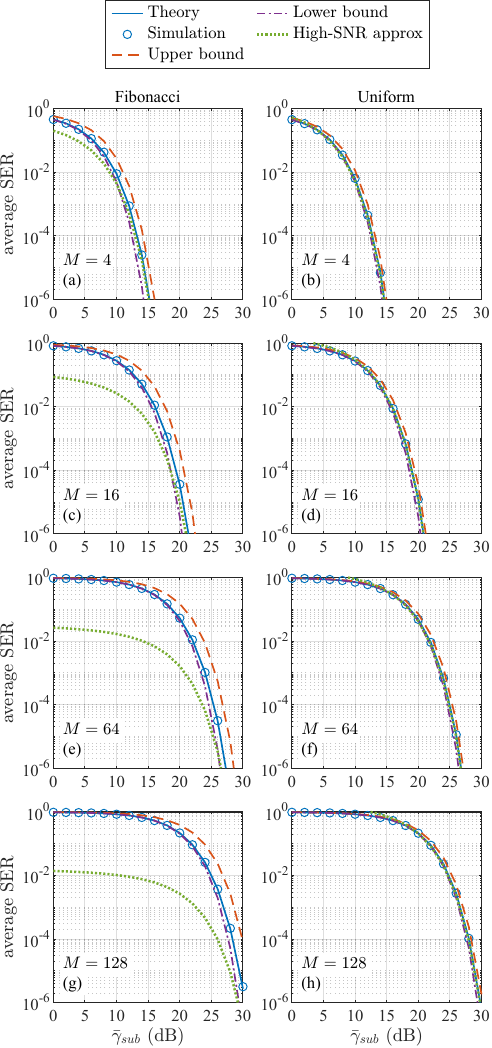}
	\caption{Average SER performance of $M$-VPM OFDM with Fibonacci-based and uniformly distributed constellations for $M \in \{4,16,64,128\}$. Theoretical results are validated by MC simulations and compared with the derived upper/lower bounds and the high-SNR approximation.}	
	\label{fig:rslt2}
\end{figure}

Fig. \ref{fig:rslt2} compares the average SER performance of the $M$-VPM DCO-OFDM system over AWGN channel using two constellation mappings: the Fibonacci-based and the uniformly distributed constellation for $M \in \{4,16,64,128\}$ (note that coordinates for uniformly distributed constellations are currently available for up to 130 points\cite{sloane2018tables}). The theoretical SER results are shown together with MC simulations, as well as the derived upper and lower bounds and the high-SNR approximation. The close agreement between the theoretical SER curves and MC simulations for both mappings verifies that the derived SER framework  is  applicable  to arbitrary  constellations. The upper and lower bounds enclose the theoretical SER curves in all cases. The bounds are noticeably tighter for the uniform constellations than for the Fibonacci-based ones. This difference arises from the constellation geometry. In a uniform mapping, the angular distances between neighboring symbols are nearly identical, so the inscribed and circumscribed cones closely approximate the actual decision boundaries. In the Fibonacci-based mapping, the decision regions are slightly irregular and the angular separations different, leading to looser bounds. The high-SNR approximation is also more accurate for the uniform constellations. When the SER falls below $10^{-1}$, the approximation almost overlaps with the theoretical SER curves. For the Fibonacci-based mapping, the approximation underestimates the SER at low SNR. This is because only part of the symbols achieve the global minimum Euclidean distance. At low SNR, errors caused by other nearby symbols become significant, hence the approximation becomes accurate only at sufficiently high SNR.

\begin{figure}[htbp]
	\centering
	\includegraphics[width=0.7\columnwidth]{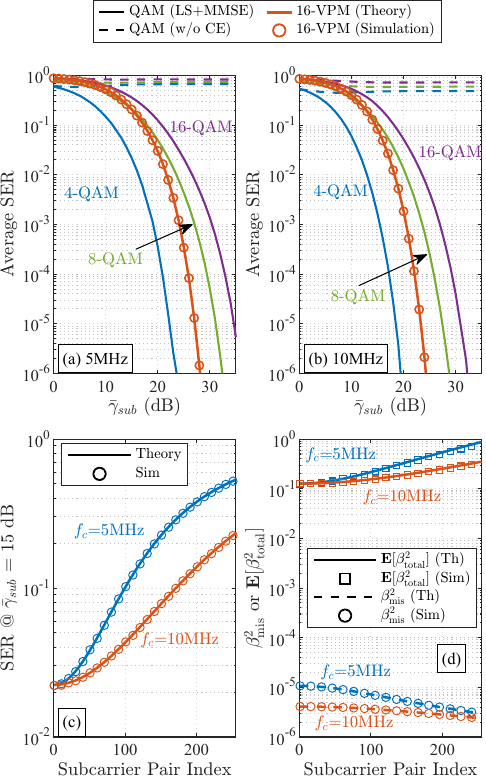}
	\caption{Average SER performance over frequency-selective OWC channels with $\tau_{\text{rms}} = 10$ ns. (a)–(b) Average SER comparison among 16-VPM, 4-QAM, 8-QAM, and 16-QAM for $f_c = 5$ MHz and $10$ MHz, respectively. (c) SER and (d) angular drift of subcarrier pairs at $\bar{\gamma}_{sub}=15$ dB.}
	\label{fig:rslt3}
\end{figure}

Fig. \ref{fig:rslt3} evaluates the DCO-OFDM system performance over frequency-selective OWC channels with $\tau_{\text{rms}}=10$ ns. Figs. \ref{fig:rslt3}(a) and (b) compare the proposed 16-VPM with 4-QAM, 8-QAM, and 16-QAM under LED cutoff frequencies $f_c=5$ MHz and $10$ MHz, respectively. For reference, QAM is evaluated via MC simulations under two configurations: with full CSI using least squares (LS) CE and minimum mean square error (MMSE) equalization (denoted as ``QAM (LS+MMSE)''), and without CE (denoted as  ``QAM (w/o CE)''). QAM without CE fails in the frequency-selective OWC channels and quickly reaches an error floor. In contrast, although 16-VPM operates without CE or equalization, it achieves lower SER than 8-QAM with the LS+MMSE receiver. Quantitatively, at a target SER of $10^{-5}$, the proposed 16-VPM provides an SNR gain of approximately $7.5$ dB over 16-QAM (LS+MMSE) and roughly $4$ dB over 8-QAM (LS+MMSE).  This result highlights the robustness of the proposed VPM scheme to frequency-selective OWC channels  and its suitability for low-latency and low-complexity links.

Comparing Figs. \ref{fig:rslt3}(a) and (b) shows that a lower cutoff frequency leads to worse average SER performance of VPM-based DCO-OFDM. The underlying reason is illustrated in Figs. \ref{fig:rslt3}(c) and (d), which plot the SER and the angular drift across subcarrier pairs at $\bar{\gamma}_{sub}=15$ dB. As shown in Fig. \ref{fig:rslt3}(c), high-frequency subcarriers suffer higher SER due to the strong low-pass attenuation of the LED. Fig. \ref{fig:rslt3}(d) explains this behavior through angular drift analysis. The channel-induced angular drift $\beta_{\text{mis}}^2$ slightly decreases at high subcarrier indices because the OWC frequency response becomes flatter in the deep stopband, reducing the relative mismatch between adjacent subcarriers. However, severe signal attenuation lowers the effective SNR, causing the noise-induced drift to dominate the total drift $\mathbb{E}[\beta_{\text{total}}^2]$. As a result, the drift increases at high indices and degrades the SER over the OFDM symbol. This effect can be mitigated by transmitter-side pre-equalization or adaptive power loading.

 \begin{figure}[htbp]
	\centering
	\includegraphics[width=0.65\columnwidth]{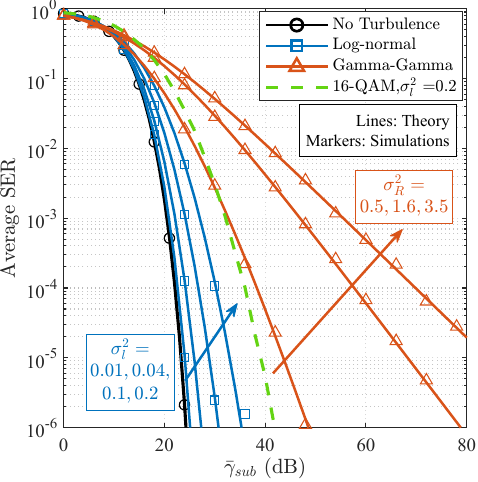}
	\caption{Average SER performance of the 16-VPM OFDM system under different turbulence conditions: no turbulence, LN fading ($\sigma_l^2 = 0.01, 0.04, 0.1, 0.2$), and GG fading ($\sigma_R^2 = 0.5, 1.6, 3.5$).}
	\label{fig:ser_turbulence}
	\label{fig:rslt5}
\end{figure}

Fig. \ref{fig:rslt5} shows the average SER performance of the proposed 16-VPM DCO-OFDM system under atmospheric turbulence. Both the LN model  ($\sigma_l^2 \in \{0.01, 0.04, 0.1, 0.2\}$) and the GG model  ($\sigma_R^2 \in \{0.5, 1.6, 3.5\}$)\cite{cao2015average} are considered. For realistic deployment conditions, the LN evaluation retains  frequency-selective dispersion ($\tau_{\text{rms}} = 10$ ns and $f_c=10$ MHz), whereas the GG channel is modeled as frequency-flat due to the much wider modulation bandwidth of typical FSO lasers compared with the  signal bandwidth of $B_w = 20$ MHz. In all cases, the theoretical curves closely match the MC simulations, confirming the accuracy of the derived analytical framework. As expected, the SER degrades as the scintillation strength ($\sigma_l^2$ or $\sigma_R^2$) increases. The slope reduction at high SNR arises from the higher probability of deep fading events. Despite severe turbulence, the proposed VPM-based DCO-OFDM still works better than the QAM-based one. For comparison, 16-QAM with LS CE and MMSE equalization is shown for the strongest LN turbulence ($\sigma_l^2 = 0.2$). 16-VPM still achieves about $6$ dB SNR gain at an SER of $10^{-5}$ compared with 16-QAM. Nevertheless, under atmospheric turbulence conditions, additional turbulence-mitigation techniques remain necessary to further improve VPM-based DCO-OFDM link reliability.

\section{Conclusion}
This paper proposed a VPM-based DCO-OFDM scheme for  OWC systems. By exploiting a rotation-based skill, a closed-form SER expression was derived for arbitrary spherical constellations over AWGN channels, together with corresponding upper and lower bounds and high-SNR approximations. The closed-form SER expression were further established under OWC frequency-selective scenarios and  atmospheric turbulence conditions. The accuracy of the developed theoretical framework was verified by extensive Monte Carlo simulations. The results show that in AWGN channels, the proposed VPM-based DCO-OFDM achieves SER performance comparable to conventional QAM-based one with the same constellation size. Under practical channel impairments, however, VPM provides clear performance advantages over QAM relying on LS CE and MMSE equalization. Because the channel effects on subcarrier pairs are largely canceled in the Stokes space, VPM effectively mitigates the degradation caused by frequency selectivity and time-varying fading. At a target SER of $10^{-5}$ in frequency-selective channels, the proposed 16-VPM achieves SNR gains of about $7.5$ dB over 16-QAM and $4$ dB over 8-QAM. Furthermore, it maintains an approximately $6$ dB gain over 16-QAM under atmospheric turbulence. These gains are obtained without CE or equalization, resulting in reduced receiver complexity and latency. The proposed VPM-OFDM scheme therefore provides a robust and low-latency solution for dynamically varying OWC systems.

\appendices

\section{Proof of Lemma \ref{lemma:inner_product}} \label{apdx:1}
The squared magnitude of the complex inner product between the received Jones vector $\mathbf{Y} = [Y_x, Y_y]^T$ and the transmitted vector $\mathbf{E}_i = [E_x, E_y]^T$ can be expressed as
\begin{equation} \label{eq:app_A1}
	|\mathbf{Y}^H \mathbf{E}_i|^2 = |Y_x|^2 |E_x|^2 + |Y_y|^2 |E_y|^2 + 2\Re(Y_x Y_y^* E_x^* E_y).
\end{equation}
According to the definitions of the unnormalized Stokes parameters, the instantaneous powers are $\hat{S}_0 = |Y_x|^2 + |Y_y|^2$ and $E_s = |E_x|^2 + |E_y|^2$. The longitudinal Stokes components are $\hat{S}_1 = |Y_x|^2 - |Y_y|^2$ and $S_{i,1} = |E_x|^2 - |E_y|^2$. The complex term are defined as $\hat{S}_2 - j\hat{S}_3 = 2Y_x Y_y^*$ and $S_{i,2} - jS_{i,3} = 2E_x E_y^*$. 

It can be obtained that the sum of $\hat{S}_0 E_s + \hat{S}_1 S_{i,1}$ as 
\begin{equation} \label{eq:app_A2}
	\hat{S}_0 E_s + \hat{S}_1 S_{i,1} = 2(|Y_x|^2 |E_x|^2 + |Y_y|^2 |E_y|^2).
\end{equation}
Utilizing the complex identity $\Re(A \cdot B^*) = \Re(A)\Re(B) + \Im(A)\Im(B)$, we have
\begin{align} \label{eq:app_A3}
	\hat{S}_2 S_{i,2} + \hat{S}_3 S_{i,3} &= \Re\left( (\hat{S}_2 - j\hat{S}_3)(\hat{S}_{i,2} - j\hat{S}_{i,3})^* \right) \nonumber \\
	&= \Re\left( (2Y_x Y_y^*)(2E_x^* E_y) \right) = 4\Re(Y_x Y_y^* E_x^* E_y).
\end{align}
Summing \eqref{eq:app_A2} and \eqref{eq:app_A3} provides the total dot product $\hat{S}_0 E_s + \hat{\mathbf{S}}^T \mathbf{S}_i = 2|\mathbf{Y}^H \mathbf{E}_i|^2$. Factoring out the total powers $\hat{S}_0 E_s$ and applying the normalization relation $\hat{\mathbf{s}}^T \mathbf{s}_i = \cos\beta$ we can achieve \eqref{eq:lemma1}. The generalized angle $\vartheta$ between two Jones vectors in $\mathbb{C}^2$ is defined by $\cos^2\vartheta = |\tilde{\mathbf{Y}}^H \mathbf{E}_i|^2 / (\|\tilde{\mathbf{Y}}\|^2 \|\mathbf{E}_i\|^2) = |\tilde{\mathbf{Y}}^H \mathbf{E}_i|^2 / (\hat{S}_0 E_s) $. Applying the trigonometric half-angle identity $\frac{1}{2}(1+\cos\beta) = \cos^2(\beta/2)$, we have $\beta = 2\vartheta$. 
\section{Proof of Lemma \ref{lemma:exact_pdf}} \label{apdx:2}
The received Jones vector $\mathbf{Y}$ follows a joint complex Gaussian distribution given in \eqref{eq:complexGaussDistri}, which can be further written as
\begin{equation} \label{eq:app_B1}
	f_{\mathbf{Y}}(\mathbf{Y}) = \frac{1}{\pi^2 N_0^2} \exp\left( -\frac{\|\mathbf{Y}\|^2 + \|\mathbf{E}_i\|^2 - 2\Re(\mathbf{Y}^H \mathbf{E}_i)}{N_0} \right).
\end{equation}
In addition, $\mathbf{Y}$ can be expressed as
\begin{equation}
	\mathbf{Y} =\begin{bmatrix} Y_x \\ Y_y \end{bmatrix} = \begin{bmatrix} Y_{x,R}+jY_{x,I}\\ Y_{y,R}+jY_{y,I} \end{bmatrix}.
\end{equation}
 In order to map those Cartesian coordinates $(Y_{x,R}, Y_{x,I}, Y_{y,R}, Y_{y,I})$ to the physical parameters, the variables are transformed into polar coordinates $Y_x = A_x e^{j\phi_x}$ and $Y_y = A_y e^{j\phi_y}$. Based on the Stokes definition, the amplitudes are expressed as $A_x = \sqrt{\hat{S}_0}\cos(\theta/2)$ and $A_y = \sqrt{\hat{S}_0}\sin(\theta/2)$. The relative phase is $\phi = \phi_y - \phi_x$, and the absolute common phase is $\phi_c = \phi_x$. 

Therefore, the Jacobian determinant for the transformation $(A_x, A_y) \to (\hat{S}_0, \theta)$ is calculated as $|J| = \frac{1}{4}$. The volume element transformation is thus defined as
\begin{align} \label{eq:app_B2}
	dV &= (A_x A_y) dA_x dA_y d\phi_x d\phi_y \nonumber \\
	&= \left( \frac{1}{2}\hat{S}_0 \sin\theta \right) \left( \frac{1}{4} d\hat{S}_0 d\theta \right) d\phi d\phi_c \nonumber \\
	&= \frac{1}{8} \hat{S}_0 d\hat{S}_0 d\Omega d\phi_c,
\end{align}
where $d\Omega = \sin\theta d\theta d\phi$ is the solid angle element. Applying the Jacobian determinant $\frac{\hat{S}_0}{8} $ to \eqref{eq:app_B1} yields the joint distribution $f(\hat{S}_0, \hat{\mathbf{s}}, \phi_c)$ as follows
\begin{equation} \label{eq:app_B_tmp1}
	f(\hat{S}_0, \hat{\mathbf{s}}, \phi_c) = \frac{\hat{S}_0}{8\pi^2 N_0^2} \exp\left( -\frac{\|\mathbf{Y}\|^2 + \|\mathbf{E}_i\|^2 - 2\Re(\mathbf{Y}^H \mathbf{E}_i)}{N_0} \right).
\end{equation}

To marginalize the unobservable common phase $\phi_c$, the term $\mathbf{Y}^H \mathbf{E}_i$ in the exponent should be explicitly expanded. Recalling \eqref{eq:vpm_generation}, the transmitted Jones vector can be written as $\mathbf{E}_i = \tilde{\mathbf{E}}_i e^{j\phi_c}$, where $\tilde{\mathbf{E}}_i$ contains the relative phase determining the Stokes parameters. Consequently, the complex inner product is expressed in polar form as
\begin{equation} \label{eq:app_B3_polar}
	\mathbf{Y}^H \mathbf{E}_i = (\mathbf{Y}^H \tilde{\mathbf{E}}_i) e^{j\phi_c} = |\mathbf{Y}^H \mathbf{E}_i| e^{j(\phi_c - \phi_{c0})},
\end{equation}
where $\phi_{c0}$ is the deterministic phase angle of the base inner product $\mathbf{Y}^H \tilde{\mathbf{E}}_i$. Taking the real part of \eqref{eq:app_B3_polar} yields $2\Re(\mathbf{Y}^H \mathbf{E}_i) = 2 |\mathbf{Y}^H \mathbf{E}_i| \cos(\phi_c - \phi_{c0})$. By substituting the amplitude equivalence $|\mathbf{Y}^H \mathbf{E}_i| = \sqrt{\frac{1}{2}\hat{S}_0 E_s (1 + \hat{\mathbf{s}}^T\mathbf{s}_i)}$ established in Lemma \ref{lemma:inner_product}, we have
\begin{equation} \label{eq:app_B3_crossterm}
	2\Re(\mathbf{Y}^H \mathbf{E}_i) = 2 \sqrt{\frac{\hat{S}_0 E_s}{2}(1 + \hat{\mathbf{s}}^T\mathbf{s}_i)} \cos(\phi_c - \phi_{c0}).
\end{equation}

Integrating \eqref{eq:app_B_tmp1} over the uniformly distributed common phase $\phi_c \in [0, 2\pi)$ marginalizes this degree of freedom. Due to the periodicity of the cosine function over a full $2\pi$ interval, the constant phase offset $\phi_{c0}$ does not affect the definite integral. Utilizing the standard integral definition of the modified Bessel function of the first kind $I_0(\cdot)$ \cite[\S9.6.16]{abramowitz1972handbook},  we have
\begin{equation} \label{eq:app_B3_final}
	f(\hat{S}_0, \hat{\mathbf{s}}) = \frac{\hat{S}_0}{4\pi N_0^2} e^{ -\frac{\hat{S}_0 + E_s}{N_0} } I_0\left( \frac{2}{N_0} \sqrt{\frac{\hat{S}_0 E_s}{2}(1 + \hat{\mathbf{s}}^T\mathbf{s}_i)} \right).
\end{equation}

Introducing  a temporary variable $u = \hat{S}_0 / N_0$ and substituting $\gamma_s = E_s / N_0$, the distribution of the spatial coordinate $\hat{\mathbf{s}}$ is obtained by integrating over $u$ as follows
\begin{equation}\label{eq:app_B4}
	f(\hat{\mathbf{s}}) = \frac{e^{-\gamma_s}}{4\pi} \int_{0}^{\infty} u e^{-u} I_0\left( 2\sqrt{u \frac{\gamma_s}{2}(1+\hat{\mathbf{s}}^T\mathbf{s}_i)} \right) du.
\end{equation}
Then, by expanding $I_0(z) = \sum_{k=0}^{\infty} \frac{(z/2)^{2k}}{(k!)^2}$ \cite[(3)]{wolfram_Bessel}, \eqref{eq:app_B4} can be further derived as
\begin{equation} \label{eq:pdf_s_hat}
	f(\hat{\mathbf{s}}) = \frac{e^{-\gamma_s}}{4\pi} \left[ 1 + \frac{\gamma_s}{2}(1 + \hat{\mathbf{s}}^T\mathbf{s}_i) \right] \exp\left({\frac{\gamma_s}{2}(1 + \hat{\mathbf{s}}^T\mathbf{s}_i)}\right).
\end{equation}

Since the rotation matrix $\mathbf{R}_i$ is orthogonal, the normalized Stokes vector transforms as $\hat{\mathbf{s}}'=\mathbf{R}_i\hat{\mathbf{s}}$ with $\mathbf{R}_i^T\mathbf{R}_i=\mathbf{I}_3$. Consequently, $\det(\mathbf{R}_i)=\pm1$, and the associated Jacobian determinant of the linear transformation is $|J|=|\det(\mathbf{R}_i)|=1$. Therefore, the differential solid-angle element remains invariant under the rotation, i.e., $d\Omega'=d\Omega$.

In the rotated spherical coordinate system, the vector $\hat{\mathbf{s}}'$ is parameterized by the zenith angle $\beta\in[0,\pi]$ and the azimuth angle $\alpha\in[0,2\pi)$. Under this rotation, the inner product simplifies to $\hat{\mathbf{s}}'^T[1,0,0]^T=\cos\beta$. Substituting this relation into \eqref{eq:pdf_s_hat} and incorporating the spherical differential element $\sin\beta$ yields the joint angular PDF in \eqref{eq:exact_pdf_beta}, thereby completing the proof of Lemma~\ref{lemma:exact_pdf}.

\section{Proof of Theorem \ref{thm:exact_ser}}\label{apdx:3}
To derive the closed-form expression of the SER, we first determine the exact boundary $\beta_{\max}^{(i)}(\alpha)$ of the decision region in the rotated spherical coordinate system.

After the rotation, the transmitted $i$-th constellation point is aligned with the zenith direction, i.e., its rotated coordinate satisfies $\beta_i = 0$. Let the spherical coordinates of an arbitrary neighboring constellation point be $(\beta_k, \alpha_k)$ for $k \ne i$. An arbitrary integration point inside the decision region is parameterized by $(\beta, \alpha)$.

According to the decision rule given in \eqref{eq:detection_metric}, the received point is correctly decoded as the $i$-th symbol if its angular distance to the zenith is smaller than its angular distance to any neighboring constellation point. By equating these angular distances, the perpendicular bisecting boundary on the sphere can be obtained as
$
	\cos\beta = \cos\beta \cos\beta_k + \sin\beta \sin\beta_k \cos(\alpha - \alpha_k).
$
After some algebraic manipulation, the decision boundary determined by the $k$-th neighbor at a given azimuth $\alpha$ can be expressed as
$
	\cot\beta = \cot\left(\frac{\beta_k}{2}\right) \cos(\alpha - \alpha_k).
$

For a point $(\beta, \alpha)$ to lie within the decision region of the $i$-th symbol, it must satisfy $\cot\beta \ge \cot\left(\frac{\beta_k}{2}\right) \cos(\alpha - \alpha_k)$ for all $k \ne i$. Since the inverse cotangent function $\mathrm{arccot}(x)$ is strictly decreasing over $(0,\pi)$, the maximum valid zenith angle $\beta_{\max}^{(i)}(\alpha)$ can be obtained by taking the maximum argument among all $M-1$ neighboring constellation points as \eqref{eq:beta_max}.

With the exact integration boundary established, the probability of correct detection $P_c^{(i)}(\gamma_s)$ is obtained by integrating the joint PDF $f_{\beta, \alpha}(\beta, \alpha | \gamma_s)$ given in \eqref{eq:exact_pdf_beta}. The inner integral with respect to $\beta$ can be analytically evaluated as
\begin{align}\label{eq:app_C_inner_integral}
	&\int_{0}^{\beta_{\max}^{(i)}(\alpha)} f_{\beta, \alpha}(\beta, \alpha | \gamma_s) \diff \beta  \nonumber\\
	& = \frac{1}{2\pi} \Bigg[ 1 - \frac{1+\cos\beta_{\max}^{(i)}(\alpha)}{2}  \exp\left(-\frac{\gamma_s}{2}\left(1-\cos\beta_{\max}^{(i)}(\alpha)\right)\right) \Bigg].
\end{align}

The symbol error rate of the $i$-th symbol is $P_e^{(i)}(\gamma_s) = 1 - P_c^{(i)}(\gamma_s)$. Since $\int_0^{2\pi} \frac{1}{2\pi} \diff\alpha = 1$, the constant term inside the bracket is canceled when subtracting from 1. Expanding the integrand directly, we have
\begin{equation}
	\begin{split}
		P_e^{(i)}(\gamma_s) &= \frac{1}{2\pi} \int_{0}^{2\pi} \frac{1+\cos\beta_{\max}^{(i)}(\alpha)}{2} \\
		&\quad \times \exp\left(-\frac{\gamma_s}{2}\left(1-\cos\beta_{\max}^{(i)}(\alpha)\right)\right) \diff\alpha.
	\end{split}
\end{equation}

To evaluate this integral using Gauss-Legendre quadrature, the interval $\alpha \in [0, 2\pi]$ is mapped to $x \in [-1, 1]$ via $\alpha = \pi x + \pi$, yielding $\diff\alpha = \pi \diff x$. Thus,
\begin{equation}\label{eq:app_C_integral_mapped}
	\begin{split}
		P_e^{(i)}(\gamma_s) &= \frac{1}{2} \int_{-1}^{1} \frac{1+\cos\beta_{\max}^{(i)}(\pi x+\pi)}{2} \\
		&\quad \times \exp\left(-\frac{\gamma_s}{2}\left(1-\cos\beta_{\max}^{(i)}(\pi x+\pi)\right)\right) \diff x.
	\end{split}
\end{equation}

Using $N_Q$ Legendre roots $x_q$ and weights $w_q$, we obtain
\begin{equation}\label{eq:app_C_integral}
	\begin{split}
		P_e^{(i)}(\gamma_s) &\approx \frac{1}{2} \sum_{q=1}^{N_Q} w_q \frac{1+\cos\beta_{\max}^{(i)}(\pi x_q+\pi)}{2} \\
		&\quad \times \exp\left(-\frac{\gamma_s}{2}\left(1-\cos\beta_{\max}^{(i)}(\pi x_q+\pi)\right)\right).
	\end{split}
\end{equation}

Under the AWGN assumption, the channel gain is constant across all subcarriers. As a result, the SER for each VPM subcarrier pair is identical, and the contribution from multiple active VPM blocks $N_v$ does not alter the average. Therefore, the closed-form average SER over an OFDM symbol reduces to the mean over the $M$ equiprobable constellation points only, as provided in Theorem~\ref{thm:exact_ser}.

\section{Proof of Theorem \ref{thm:high_snr}} \label{apdx:4}
The additive complex Gaussian noise is circularly symmetric, i.e., its distribution remains unchanged under any unitary rotation in the 2D Jones space. Through the Jones-to-Stokes mapping, such unitary transformations correspond to 3D rigid-body rotations on the Poincaré sphere. This implies that the noise statistics are rotationally invariant in the Stokes space. As a result, the local noise covariance derived at any given point is statistically equivalent and geometrically valid for all constellation points.

Therefore, without loss of generality, assume the transmitted constellation point is $\mathbf{S}_0 = [E_s, 0, 0]^T$, corresponding to the ideal Jones vector $\mathbf{E}_0 = [\sqrt{E_s}, 0]^T$. Let the complex noise vector be $\mathbf{Z} = [Z_{xr} + jZ_{xi}, Z_{yr} + jZ_{yi}]^T$, where the four real orthogonal components are independent and identically distributed with variance $N_0/2$.

Then, the received VPM block's Stokes parameters are computed according to \eqref{eq:stokes_def} . 
In the high-SNR regime ($E_s \gg N_0$),  we only retain the dominant first-order  noise terms gives the equivalent 3D noise vector $\Delta \mathbf{S}$:
\begin{equation} \label{eq:app_D2}
	\begin{bmatrix} \hat{S}_1 \\ \hat{S}_2 \\ \hat{S}_3 \end{bmatrix} \approx 
	\begin{bmatrix} E_s \\ 0 \\ 0 \end{bmatrix} + 
	\underbrace{\begin{bmatrix} 2\sqrt{E_s}Z_{xr} \\ 2\sqrt{E_s}Z_{yr} \\ -2\sqrt{E_s}Z_{yi} \end{bmatrix}}_{\Delta\mathbf{S}}.
\end{equation}

Since $Z_{xr}, Z_{yr},$ and $Z_{yi}$ are mutually independent, the three orthogonal dimensions of $\Delta\mathbf{S}$ are independent. The variance of each dimension equals $(2\sqrt{E_s})^2(N_0/2) = 2E_s N_0$. Consequently, the equivalent noise in the Stokes space asymptotically follows an isotropic multivariate Gaussian distribution, i.e., $\Delta\mathbf{S} \sim \mathcal{N}(\mathbf{0}, \sigma_s^2 \mathbf{I}_3)$, where $\sigma_s^2 = 2E_s N_0$.

To evaluate the pairwise error probability (PEP) between $\mathbf{S}_0$ and a neighboring point separated by the shortest chord length $d_{min} = 2E_s^2(1-\cos\beta_{min})$, the 3D isotropic Gaussian noise is projected onto the 1D axis connecting those two points. By a fundamental property of isotropic multivariate Gaussians, the projection along any axis is a 1D Gaussian with the same variance $\sigma_s^2$.

The optimal decision boundary between the two points is the perpendicular bisector located at $d_{min}/2$. The PEP is thus formulated as the tail integral of this 1D Gaussian distribution:
\begin{equation} \label{eq:app_D3}
	PEP = Q\left( \frac{d_{min}/2}{\sigma_s} \right)= Q\left( \sqrt{\frac{\gamma_s}{4}(1-\cos\beta_{min})} \right).
\end{equation}
Applying the concept of union bound \cite{proakis2008digital} over the $\bar{N}_{min}$ nearest neighbors completes the derivation of the high-SNR approximation.

\section{Proof of Theorem \ref{thm:exact_mismatch_ser}}\label{apdx:E}
To derive the average SER under frequency-selective channels, we adopt a spherical integration approach similar to that in Appendix C. Due to the channel-induced angular drift, the probability density center of the $i$-th symbol on the $m$-th subcarrier block shifts to a new coordinate $\mathbf{s}'_{m,i}$.

During manipulating the integration in $P_{c,m}^{(i)}$, a subcarrier-specific rotation matrix $\mathbf{R}'_{m,i}$ is applied to align the drifted center $\mathbf{s}'_{m,i}$ with the zenith. The ideal decision hyperplanes separating the constellation points $\mathbf{s}_i$ and $\mathbf{s}_j$ are rotated accordingly. Since the normal vector of the hyperplane is given by the difference vector of the two points, the updated normal vector becomes $\mathbf{w}'_{m,i,j} = \mathbf{R}'_{m,i} (\mathbf{s}_i - \mathbf{s}_j) = [w'_{j,x}, w'_{j,y}, w'_{j,z}]^T$.

The rotated plane equation is $\mathbf{s}_{\text{rot}}^T \mathbf{w}'_{m,i,j} = 0$. Substituting the spherical representation $\mathbf{s}_{\text{rot}} = [\sin\beta\cos\alpha, \sin\beta\sin\alpha, \cos\beta]^T$ gives
\begin{equation}
	w'_{j,z} \cos\beta = - \sin\beta \left( w'_{j,x} \cos\alpha + w'_{j,y} \sin\alpha \right).
\end{equation}
Dividing by $w'_{j,z} \sin\beta$ isolates the cotangent of the zenith boundary angle. Since the valid region is the intersection of the half-spaces defined by all neighbors, the effective boundary corresponds to the smallest $\beta$. Because $\mathrm{arccot}(x)$ is strictly decreasing on $(0,\pi)$, this is equivalent to maximizing the argument over all $M-1$ neighbors, resulting in \eqref{eq:beta_m_f}.

With the boundary $\beta_m^{(i)}(\alpha)$ determined for the $m$-th subcarrier block, the integration procedure in Appendix C can be applied directly. Substituting $\beta_m^{(i)}(\alpha)$ into \eqref{eq:app_C_inner_integral}, performing the Gauss–Legendre substitution ($\alpha = \pi x_q + \pi$), and averaging over all $N_v$ valid subcarrier blocks and $M$ equiprobable constellation symbols yields the closed-form SER expression in Theorem~\ref{thm:exact_mismatch_ser}.

\bibliographystyle{IEEEtran}
\bibliography{IEEEabrv,myReferences.bib}

\vfill

\end{document}